\documentclass[11pt,letterpaper]{article}

\usepackage{sectsty}
\usepackage{notoccite}
\usepackage{graphics} 
\usepackage{amsmath} 
\usepackage{color,xspace}
\usepackage{graphicx}
\usepackage{subcaption}
\usepackage{cite}
\usepackage{cleveref}
\usepackage{bm}
\usepackage{bbm}
\usepackage{tikz}
\usepackage{epstopdf} 
\usepackage{enumitem}
\usepackage{mdframed}   
\usepackage{sectsty}
\usepackage{amsthm}
\usepackage{amsfonts}

\usepackage[margin=1.0in]{geometry}

\usepackage{caption} 
\usepackage{caption} 
\newtheorem{theorem}{Theorem}
\newtheorem{lemma}{Lemma}    

\newtheorem{remark}{Remark}
\newtheorem{example}{Example}    

\newtheorem{proposition}{Proposition}    
    
\newtheorem{definition}{Definition}

\title{{\LARGE \bf Stability Analysis of Complementarity Systems with Neural Network Controllers}}

\author{Alp Aydinoglu \thanks{Department of Electrical and Systems Engineering, University of Pennsylvania. Email: \{alpayd, morari, posa\}@seas.upenn.edu} \and Mahyar Fazlyab  \thanks{Mathematical Institute for Data Science, Johns Hopkins University. Email: mahyarfazlyab@jhu.edu} \and Manfred Morari \footnotemark[1] \and Michael Posa \footnotemark[1]
}

\date{}

\begin{document}
\pagestyle{plain}
\maketitle

\begin{abstract}
	Complementarity problems, a class of mathematical optimization problems with orthogonality constraints, are widely used in many robotics tasks, such as locomotion and manipulation, due to their ability to model non-smooth phenomena (e.g., contact dynamics). In this paper, we propose a method to analyze the stability of complementarity systems with neural network controllers. First, we introduce a method to represent neural networks with rectified linear unit (ReLU) activations as the solution to a linear complementarity problem. Then, we show that systems with ReLU network controllers have an equivalent linear complementarity system (LCS) description. Using the LCS representation, we turn the stability verification problem into a linear matrix inequality (LMI) feasibility problem. We demonstrate the approach on several examples, including multi-contact problems and friction models with non-unique solutions.
\end{abstract}

\section{Introduction}


Due to recent advancements in deep learning, there has been an increasing interest in using neural networks (NNs) to stabilize dynamical systems. For instance, neural networks have been used to approximate model predictive control policies through supervised learning \cite{parisini1995receding,zhang2019safe,karg2020efficient,hertneck2018learning}, or reinforcement learning \cite{chen2018approximating}. 
Although neural network controllers can achieve satisfactory performance under less restrictive assumptions about the model of the dynamical system or the environment it operates in, they lack guarantees. This drawback limits the application of neural networks in safety-critical systems, in which simpler control strategies, although potentially inferior to deep neural networks in performance, do have performance guarantees. Therefore, it is critical to develop tools that can provide useful certificates of stability, and robustness for NN-driven systems.

Many important robotics systems are non-smooth and researchers have shown the effectiveness of NN policies \cite{xie2018feedback, haarnoja2018learning, zhu2019dexterous} on such systems without providing formal guarantees.
The goal of this paper is to introduce a method for stability analysis of non-smooth systems in feedback loops with NN controllers.
Our framework is inspired by complementarity systems \cite{heemels2000linear}, differential equations coupled with the solution of a linear complementarity problem. Complementarity problems are a class of mathematical optimization problems with orthogonality constraints \cite{cottle2009linear}. Linear complementarity problems, in particular, are widely used in computational non-smooth mechanics with unilateral contacts and friction \cite{brogliato1999nonsmooth}, and more generally, in applications involving quadratic programming \cite{o2020operator}. In simple terms, a linear complementarity problem can be stated as the following potentially non-convex quadratic optimization problem,
\begin{align*}
	\operatorname{minimize} \ \lambda^\top (F \lambda + q) \quad \text{subject to } F\lambda + q \geq 0, \ \lambda \geq 0.
\end{align*}
With the objective function being non-negative, the solutions to the optimization problem satisfy the complementarity condition $(F\lambda + q)^\top \lambda = 0$.
In the context of contact dynamics, for example, one can interpret $\lambda$ as a contact force between a robot and a surface, and $F \lambda + q$ is a gap function relating the contact force and the distance from the robot to the contact surface.
Because of their ability to model set-valued and non-smooth functions, complementarity problems are widely used within the robotics community, particularly to simulate contact dynamics \cite{stewart1996implicit, halm2018quasi}, leveraged in trajectory optimization \cite{posa2014direct}, and stability analysis of rigid-body systems with contacts \cite{posa2015stability, camlibel2007lyapunov, aydinoglu2020contact}.

\subsection{Related Work}

The connection between nonlinearities in neural networks and mathematical optimization has been exploited recently in various contexts. 
%
In \cite{raghunathan2018semidefinite,fazlyab2019safety,fazlyab2019probabilistic} the authors use quadratic constraints to describe ReLU activation functions followed by a semidefinite relaxation to perform robustness 
analysis of ReLU networks. In \cite{fazlyab2019efficient}, the authors exploit the fact that all commonly used activation functions in deep learning are gradients of convex potentials, hence they satisfy incremental quadratic constraints that can be used to bound the global Lipschitz constant of feed-forward neural networks.  The work in \cite{amos2017optnet} integrates quadratic programs as end-to-end trainable deep networks to encode constraints and more complex dependencies  between the
hidden states. Yin et al. \cite{yin2020stability} considers uncertain linear time-invariant systems with neural network controllers. By over approximating the input-output map of the neural network and uncertainties by quadratic and integral quadratic constraints, respectively, the authors develop an SDP whose solution yields quadratic Lyapunov functions.
In \cite{chen2020learning} the authors develop a learning-based iterative sample guided strategy based on the analytic center cutting plane method to search for Lyapunov functions for piece-wise affine systems in feedback with ReLU networks where the generation of samples relies on solving mixed-integer quadratic programs. In \cite{karg2020stability}, the authors use
a mixed-integer linear programming formulation to perform output range analysis of ReLU neural networks and provide guarantees for constraint satisfaction and asymptotic stability of the closed-loop system.

\subsection{Contributions}
Inspired by the connection between ReLU functions and linear complementarity problems, we develop a method to analyze linear complementarity systems in feedback with ReLU network controllers. Our starting point is to show that a single ReLU activation can be expressed as the solution to a linear complementarity problem (Lemma \ref{lemma:LCP_max_eqivalency}). Using this, we show that we can represent ReLU neural networks as linear complementarity problems (Lemma \ref{LCP_ReLU_equivalency}).
Next, we demonstrate that linear complementarity systems with neural network controllers have an equivalent LCS representation. Then, we leverage the theory of stability analysis for complementarity systems and derive the discrete time version of the results in \cite{camlibel2007lyapunov}. We describe the sufficient conditions for stability in the form of Linear Matrix Inequalities (LMI's). Denoting by $N$ the number of neurons in the network plus the number of complementarity variables in the LCS, the size of the LMI's scales linearly with $N$. Furthermore, the maximum possible number of decision variables in our LMI scales quadratically with $N$. To the best of our knowledge, this is the first work on analyzing the stability of LCS systems with neural network controllers.

\section{Background}

\subsection{Notation}

We denote the set of non-negative integers by $\mathbb{N}_0$, the set of d-dimensional vectors with real components as $\mathbb{R}^d$ and the set of $n \times m$ dimensional matrices by $\mathbb{R}^{n \times m}$. 
For two vectors $a \in \mathbb{R}^m$ and $b \in \mathbb{R}^m$, we use the notation $0 \leq a \perp b \geq 0$ to denote that $a \geq 0, \; b \geq 0, \; a^T b = 0$.
For a positive integer $l$, $\bar{l}$ denotes the set $\{ 1, 2, \ldots, l \}$. 
Given a matrix $M \in \mathbb{R}^{k \times l}$ and two subsets $I \subseteq \bar{k}$ and $J \subseteq \bar{l}$, we define $M_{I J} = (m_{i j})_{i \in I, j \in J}$. 
For the case where $J = \bar{l}$, we use the shorthand notation $M_{I \bullet}$. 

\subsection{Linear Complementarity Problem}

The theory of linear complementarity problems (LCP) will be used throughout this work \cite{cottle2009linear}.
\begin{definition} \label{def: LCS_definition}
	Given a vector $q \in \mathbb{R}^m$, and a matrix ${F \in \mathbb{R}^{m \times m}}$, the $LCP(q,F)$  describes the following mathematical program:
	\begin{alignat}{2}
		\label{LCS_definiton}
		\notag & \underset{}{\text{find}} && \lambda \in \mathbb{R}^m \\
		\notag & \text{subject to}  \quad && y = F \lambda + q,\\
		\notag & \quad && 0 \leq \lambda \perp y \geq 0.
	\end{alignat}
\end{definition}
The solution set of the $LCP(q,F)$ is denoted by
\begin{equation*}
	\text{SOL}(q,F) = \{ \lambda : y = F \lambda+q, 0 \leq \lambda \perp y \geq 0 \}.
\end{equation*}
The $\text{LCP}(q,F)$ may have multiple solutions or none at all. 
The cardinality of the solution set $\text{SOL}(q,F)$ depends on the matrix $F$ and the vector $q$. In particular, if $F$ is a P-matrix, $\text{SOL}(q,F)$ is always a singleton.
\begin{definition}
	A matrix $F \in \mathbb{R}^{m \times m}$ is a P-matrix, if the determinant of all of its principal sub-matrices are positive; that is, $\text{det}(F_{\alpha \alpha}) > 0$ for all $\alpha \subseteq \{ 1, \ldots, m \} $.
\end{definition}
The solution set $\text{SOL}(q,F)$ is a singleton for all $q$ if $F$ is a P-matrix \cite{cottle2009linear}.
If we denote the unique element of $\text{SOL}(q,F)$ as $\lambda(q)$, then $\lambda(q)$ is a piece-wise linear function of $q$. 
We can describe this function explicitly as in \cite{camlibel2007lyapunov}. Consider $y = F \lambda(q) + q$, and define the index sets
\begin{equation*}
	\alpha(q) = \{ i : \lambda_i (q) > 0 = y_i \},
\end{equation*}
\begin{equation*}
	\beta(q) = \{ i : \lambda_i (q) = 0 \leq y_i \},
\end{equation*}
Then, $\lambda(q)$ is equivalent to
\begin{equation}
	\label{eq:piecewise_affine_rep}
	\lambda_{\alpha}(q) = - ( F_{\alpha \alpha} )^{-1} I_{\alpha \bullet} q, \; \; \lambda_{\beta} (q) = 0,
\end{equation}
where $\alpha = \alpha(q)$ and $\beta = \beta(q)$. Furthermore, $\lambda(q)$ as in \eqref{eq:piecewise_affine_rep} is Lipschitz continuous since it is a continuous piece-wise linear function of $q$ \cite{scholtes2012introduction}.

\subsection{Linear Complementarity Systems}
We are now ready to introduce linear complementarity systems (LCS). 
In this work, we consider an LCS as a difference equation coupled with a variable that is the solution of an LCP.

\begin{definition}
	A linear complementarity system describes the  trajectories $( x_k )_{k \in \mathbb{N}_{0}}$ and $( \lambda_k )_{k \in \mathbb{N}_{0}}$ for an input sequence $( u_k )_{k \in \mathbb{N}_{0}}$ starting from $x_0$ such that
	\begin{equation}
		\label{eq:LCS}
		\begin{aligned}
			& x_{k+1} = A x_k + B u_k + D\lambda_k + z,\\
			& y_k = Ex_k +  F \lambda_k + H u_k + c, \\
			& 0 \leq \lambda_k \perp y_k \geq 0.
		\end{aligned}
	\end{equation}
	where $x_k \in \mathbb{R}^{n_x}$, $\lambda_k \in \mathbb{R}^{n_{\lambda}}$, $u_k \in \mathbb{R}^{n_u}$, $A \in \mathbb{R}^{n_x \times n_x}$, $B \in \mathbb{R}^{n_x \times n_u}$, $D \in \mathbb{R}^{n_x \times n_\lambda}$, $z \in \mathbb{R}^{n_x}$, $E \in \mathbb{R}^{n_{\lambda} \times n_x}$, $F \in \mathbb{R}^{n_x \times n_\lambda}$, $H \in \mathbb{R}^{n_\lambda \times n_u}$ and $c \in \mathbb{R}^{n_\lambda}$.
\end{definition}

For a given $k$, $x_k$ and $u_k$, the corresponding complementarity variable $\lambda_k$ can be found by solving $\text{LCP}(E x_k + H u_k + c, F)$ (see Definition \ref{def: LCS_definition}). Similarly, $x_{k+1}$ can be computed using the first equation in \eqref{eq:LCS} when $x_k, u_k$ and $\lambda_k$ are known. 
In general, the trajectories $ (x_k)$ and $(\lambda_k)$ are not unique since $\text{SOL}(Ex_k + H u_k + c, F)$ can have multiple elements; hence, \eqref{eq:LCS} is a difference \emph{inclusion} \cite{dontchev1992difference}.

In this work, we will focus on autonomous linear complementarity systems (A-LCS) because we consider the input as a function of the state and the complementarity variable, i.e., $u = u(x_k, \lambda_k)$. An A-LCS represents the evolution of trajectories $( x_k )_{k \in \mathbb{N}_{0}}$ and $( \lambda_k )_{k \in \mathbb{N}_{0}}$ according to following dynamics,
\begin{equation}
	\label{eq:auto_LCS}
	\begin{aligned}
		& x_{k+1} = Ax_k + D\lambda_k + z,\\
		& y_k = Ex_k +  F \lambda_k + c, \\
		& 0 \leq \lambda_k \perp y_k \geq 0,
	\end{aligned}
\end{equation}
and unlike \eqref{eq:LCS} there is no input. Moving forward, we will consider A-LCS models that can have non-unique trajectories. 

We note that, however, the existence of a special case of \eqref{eq:auto_LCS} is continuous piecewise affine systems \cite{heemels2001equivalence}. If $F$ is a P-matrix, then $\lambda(x_k)$ is unique for all $x_k$ and \eqref{eq:auto_LCS} is equivalent to
\begin{equation*}
	\begin{aligned}
		& x_{k+1} = Ax_k + B u_k + D\lambda(x_k) + z,
	\end{aligned}
\end{equation*}
where $\lambda(x_k)$ is the unique element of $\text{SOL}(Ex_k + c, F)$ and can be explicitly described as in \eqref{eq:piecewise_affine_rep}. 
In this setting, \eqref{eq:LCS} is a piece-wise affine dynamical system and has a unique solution for any initial condition $x_0$.

\subsection{Stability of A-LCS}

We introduce the notions of stability for A-LCS that are similar to \cite{smirnov2002introduction}. An equilibrium point $x_e$ for \eqref{eq:auto_LCS} is defined as a point that satisfies $x_e = A x_e + D \lambda_e + z$ where $\lambda_e = \text{SOL}(Ex_e+c,F)$ is a singleton.
Without loss of generality, we assume $x_e = 0$ is an equilibrium of the system, i.e., $D \text{SOL}(c, F) = \{ -z \}$.

\begin{definition}
	The equilibrium $x_e = 0$ of A-LCS is
	\begin{enumerate}
		\item stable if for any given $\epsilon > 0$, there exists a $\delta > 0$ such that
		\begin{equation*}
			||x_0|| \leq \delta \implies ||x_k|| \leq \epsilon \; \forall k \geq 0,
		\end{equation*}
		for any trajectory $\{x_k\}$ starting from $x_0$,
		\item asymptotically stable if it is stable and there is a $\delta > 0$ such that
		\begin{equation*}
			||x_0|| \leq \delta \implies \lim_{k \rightarrow \infty} ||x_k|| = 0,
		\end{equation*}
		for any trajectory $\{x_k\}$ starting from $x_0$.
		\item geometrically stable if there exists $\delta > 0$, $\alpha > 1$ and $0 < \rho <1$ such that
		\begin{equation*}
			||x_0|| \leq \delta \implies ||x_k|| \leq \alpha \rho^k ||x_0|| \; \forall k \geq 0,
		\end{equation*}
		for any trajectory $\{x_k\}$ starting from $x_0$.
	\end{enumerate}
\end{definition}
Notice that if $F$ is a P-matrix, these are equivalent to the notions of stability for difference equations where the right side is Lipschitz continuous \cite{khalil2002nonlinear} since there is a unique trajectory $\{x_k\}$ starting from any initial condition $x_0$.

\section{Linear Complementarity Systems with Neural Network Controllers}
\label{sec:comp_NN}
In this section, we demonstrate that neural networks with rectified linear units (ReLU) have an equivalent LCP representation. 
Then, we show that an LCS combined with a neural network controller has an alternative complementarity system description.
\begin{definition}
	A ReLU neural network $\phi \colon \mathbb{R}^{n_x} \mapsto \mathbb{R}^{n_\phi}$ with $L$ hidden layers is the composite function
	\begin{equation}
		\label{eq:NN_ReLU}
		\phi(x) = (h_L \circ \lambda_\text{ReLU} \circ h_{L-1} \circ \ldots \circ \lambda_\text{ReLU} \circ h_0)(x),
	\end{equation}
	where $\lambda_\text{ReLU}(x)~=~\max \{0, x \}$ is the ReLU activation layer, and $h_i(x) = \theta_i x + c_i$ are the affine layers with $\theta_i \in \mathbb{R}^{n_{i+1} \times n_i}$, $c_i \in \mathbb{R}^{n_{i+1}}$. Here, $n_{i}, 1 \leq i \leq L$ denotes the number of hidden neurons in the $i$-th layer, $n_0 = n_x$, and $n_{L+1} = n_{\phi}$. We denote by $n_t = \sum_{i=1}^{L} n_i$ the total number of neurons.
\end{definition}

\subsection{Representing ReLU Neural Networks as Linear Complementarity Problems}

ReLU neural networks are piece-wise affine functions. Similarly, the linear complementarity problem describes a piece-wise affine function as shown in \eqref{eq:piecewise_affine_rep} as long as $F$ is a P-matrix. In this section, we will explore the connection between two piece-wise affine representations.

It has been shown that ReLU neural networks can be represented with quadratic constraints \cite{raghunathan2018semidefinite}, \cite{fazlyab2019safety}. Now, we will show the connection between these results and linear complementarity problems. Our goal is to describe a method to represent a multi-layered ReLU neural network as a linear complementarity problem.

First consider a single ReLU unit $\lambda_{\text{ReLU}}(x) = \max(0,x)$ and show its equivalent LCP representation.
\begin{lemma}
	\label{lemma:LCP_max_eqivalency}
	Consider the following LCP for a given $x \in \mathbb{R}^d$:
	\begin{alignat*}{2}
		\notag & \underset{}{\text{find}} && \lambda^{\text{LCP}} \in \mathbb{R}^d \\
		\notag & \text{subject to} \quad && \bar{y} = -x + \lambda^{\text{LCP}}, \\
		\notag & && 0 \leq \lambda^\text{LCP} \perp \bar{y} \geq 0,
	\end{alignat*}
	Then $\lambda^\text{LCP}$ is unique and is given by $\lambda^\text{LCP}= \max \{0, x \}$.
\end{lemma}
\begin{proof}
	If $x_i < 0$, then $\lambda^\text{LCP}_i = 0$ and if $x_i \geq 0$, then $\lambda^\text{LCP}_i = x_i$.
\end{proof}
Next, we consider a two layered neural network and transform it into an LCP using Lemma \ref{lemma:LCP_max_eqivalency}.
\begin{figure}[t!]
	\centering
	\label{fig:two_layered}
	\includegraphics[width=0.4\columnwidth]{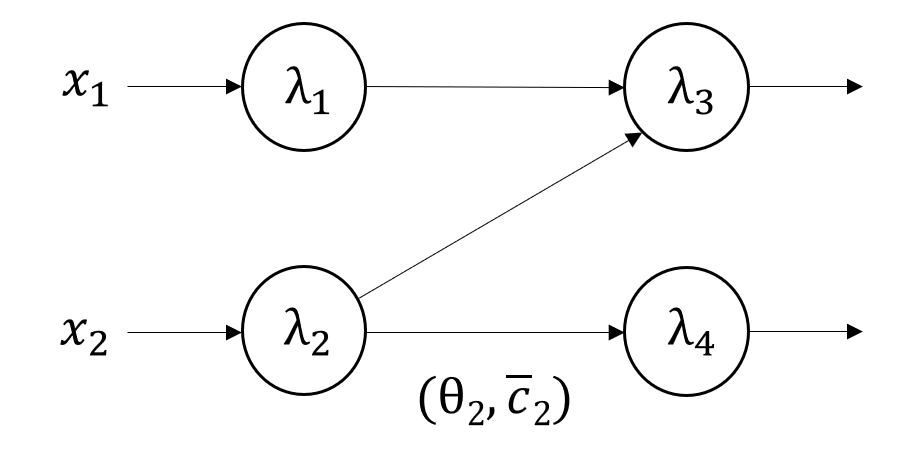}
	\caption{Two-layered neural network with ReLU activation functions.}
	\label{fig:two_layered_NN}
\end{figure}

\begin{example}
	Consider a two layered network shown in Figure \ref{fig:two_layered_NN} where $\phi_\text{2-layer}(x) = \lambda_{\text{ReLU}} \circ h_1 \circ \lambda_{\text{ReLU}} \circ h_0 (x)$ with $h_0(x) = x$ and $h_1(x) = \theta_2 x + \bar{c}_2$. 
	An alternative representation is $\phi_\text{2-layer}(x)$ $=\begin{bmatrix}
		\lambda_3(x)^\top & \lambda_4(x)^\top 
	\end{bmatrix}^\top$ where
	\begin{align}
		\label{eq:lambda1}
		& \lambda_1(x) = \max\{ 0, x_1 \}, \\
		\label{eq:lambda2}
		& \lambda_2(x) = \max\{ 0, x_2 \}, \\
		\label{eq:lambda3}
		& \lambda_3(x) = \max\{ 0, \theta_2^{1,1} \lambda_1(x) + \theta_{2}^{1,2} \lambda_2(x) + \bar{c}_2^1\}, \\
		\label{eq:lambda4}
		& \lambda_4(x) = \max\{ 0, \theta_{2}^{2,2} \lambda_2(x) + \bar{c}_2^2 \}.
	\end{align}
	Here $\lambda_i$ represents the output of the $i$th ReLU activation function, $\theta_{i}^{j,k}$ and $\bar{c}_i^j$ represent the coefficients of the affine function.
	Observe that the two-layered NN is equivalent to $\begin{bmatrix}
		\lambda_3 \\ \lambda_4
	\end{bmatrix}$ where $\lambda$ is the unique solution of the following LCP:
	\begin{alignat*}{2}
		\notag & \underset{}{\text{find}} && \lambda \in \mathbb{R}^4 \\
		& \text{subject to} \quad &&  \bar{y}_1 = -x_1 + \lambda_1, \\
		& \quad && \bar{y}_2 = -x_2 + \lambda_2, \\
		& \quad && \bar{y}_3 = -\theta_2^{1,1} \lambda_1 - \theta_{2}^{1,2} \lambda_2 - \bar{c}_2^1 + \lambda_3, \\
		& \quad && \bar{y}_4 = -\theta_{2}^{2,2} \lambda_2 - \bar{c}_2^2 + \lambda_4, \\
		& \quad && 0 \leq \lambda_1 \perp \bar{y}_1 \geq 0, \\
		& \quad && 0 \leq \lambda_2 \perp \bar{y}_2 \geq 0, \\
		& \quad && 0 \leq \lambda_3 \perp \bar{y}_3 \geq 0, \\
		& \quad && 0 \leq \lambda_4 \perp \bar{y}_4 \geq 0. \\
	\end{alignat*}
	Here, $\{ \lambda_i \}_{i=1}^2$ can be represented as $\lambda_i = \max \{0, x_i \}$ and $\{ \lambda_i \}_{i=3}^4$ are as in \eqref{eq:lambda3}, \eqref{eq:lambda4} after direct application of Lemma \ref{lemma:LCP_max_eqivalency}. Then, we conclude that $\begin{bmatrix}
		\lambda_3 \\ \lambda_4
	\end{bmatrix} = \phi_\text{2-layer}$.
\end{example}
Now, we show that all neural networks of the form \eqref{eq:NN_ReLU} have an equivalent LCP representation.
%
%
\begin{lemma}
	\label{LCP_ReLU_equivalency}
	For any $x$, the ReLU neural network in \eqref{eq:NN_ReLU} can be expressed as $\phi(x) =  \bar{D} \lambda(x) + \bar{z}$, where $\lambda(x)$ is the unique solution of the following linear complementarity problem:
	\begin{alignat}{2}
		\notag & \underset{}{\text{find}} && \lambda\\
		\notag & \text{subject to}  \quad && \bar{y} = \bar{E} x + \bar{F} \lambda + \bar{c},\\
		\notag & \quad && 0 \leq \lambda \perp \bar{y} \geq 0,
	\end{alignat}
	where 
	$\bar{c} = \begin{bmatrix}
		-c_0 \\
		-c_1 \\
		\vdots \\
		-c_{L-1}
	\end{bmatrix}, \; \bar{E} = \begin{bmatrix}
		-\theta_0 \\
		0 \\
		\vdots \\
		0
	\end{bmatrix}$, $\bar{F} = \begin{bmatrix}
			I & \; & \; & \;  \\
			-\theta_1 & I  & \; & \;\\
			0 & -\theta_2  & I & \;\\
			0 & 0  & -\theta_3 & I \\
			\vdots & \vdots & \vdots & \vdots & \vdots & \vdots \\
			0 & \ldots & \ldots & \ldots & -\theta_{L-1} & I
		\end{bmatrix}$,

	and $\bar{z} = \bar{c}_L$, $\bar{D} = \begin{bmatrix}
		0 & 0 &\ldots & 0 & \theta_L
	\end{bmatrix}$ where $\bar{F}$ is a P-matrix. 
\end{lemma}
\begin{proof}
	First, we write $\lambda^\top = \begin{bmatrix}
		\lambda_0^\top &
		\lambda_1^\top &
		\cdots & 
		\lambda_{L-1}^\top
	\end{bmatrix} \in \mathbb{R}^{1 \times n_t}$ where each sub vector $\lambda_i$ has the same dimension as $\bar{c}_i$. 
	Next, we show that $\lambda_0(x) = \lambda_\text{ReLU} \circ h_0(x)$. 
	Observe that $\lambda_0$ is independent of $\lambda_1, \ldots, \lambda_{L-1}$ since $F$ is lower-triangular. 
	Hence $\lambda_0$ is the unique element of the following LCP:
	\begin{alignat}{2}
		\notag & \underset{}{\text{find}} && \lambda_0\\
		\notag & \text{subject to}  \quad && \bar{y}_0 = -\theta_0 x - \bar{c}_0 + \lambda_0,\\
		\notag & \quad && 0 \leq \lambda_0 \perp \bar y_0 \geq 0.
	\end{alignat}
	Following Lemma \ref{lemma:LCP_max_eqivalency}, $\lambda_0(x) = \max \{ 0,h_0(x) \} = \lambda_\text{ReLU} \circ h_0(x)$. Similarly, notice that for $i > 0$, $\lambda_i$ only depends on $\lambda_{i-1}(x)$ and is the unique element of:
	\begin{alignat}{2}
		\notag & \underset{}{\text{find}} && \lambda_i\\
		\notag & \text{subject to}  \quad && \bar{y}_i = -\theta_i \lambda_{i-1}(x) - \bar{c}_i + \lambda_i,\\
		\notag & \quad && 0 \leq \lambda_i \perp \bar{y} \geq 0,
	\end{alignat}
	and is equivalent to $\lambda_i(x) = \lambda_\text{ReLU} \circ h_i \circ \lambda_{i-1} (x)$ as a direct application of Lemma \ref{lemma:LCP_max_eqivalency}. 
	Using this equivalency recursively,
	\begin{equation*}
		\bar{D} \lambda(x) + \bar{z} = h_L \circ \lambda_\text{ReLU} \circ h_{L-1} \circ \ldots \circ \lambda_\text{ReLU} \circ h_0(x).
	\end{equation*}
	Notice that $\bar{F}_{\alpha \alpha}$ is lower triangular with ones on the diagonal for any $\alpha$ such that $\text{card}(\alpha) \geq 2$ hence $\bar{F}$ is a P-matrix.
\end{proof}
Each neuron in the NN is represented with a complementarity variable, therefore the dimension of the complementarity vector ($\lambda$) is equal to the number of neurons in the network. As seen in Lemma \ref{LCP_ReLU_equivalency}, transforming a ReLU neural network into an LCP only requires concatenating vectors and matrices.

\subsection{Linear Complementarity Systems with Neural Network Controllers}
\label{sub:complementaritysys_nncont}

We will use the LCP representation of the neural network \eqref{eq:NN_ReLU} and describe an LCS with a NN controller as an A-LCS. Consider a linear complementarity system with a ReLU neural network controller $u_k = \phi(x_k)$:
\begin{equation}
	\begin{aligned}
		\label{eq:LCS_NN_cont}
		& x_{k+1} = A x_k + B \phi(x_k) + \tilde{D} \tilde{\lambda}_k + \tilde{z},\\
		& \tilde{y}_k = \tilde{E} x_k +  \tilde{F} \tilde{\lambda}_k + H \phi(x_k) + \tilde{c}, \\
		& 0 \leq \tilde{\lambda}_k \perp \tilde{y}_k \geq 0,
	\end{aligned}
\end{equation}
where $x_k \in \mathbb{R}^{n_x}$ is the state, $\tilde{\lambda}_k \in \mathbb{R}^{n_{\tilde \lambda}}$ is the complementarity variable, $\phi(x) \in \mathbb{R}^{n_\phi}$ is a ReLU neural network as in \eqref{eq:NN_ReLU} with $n_t$ neurons. Notice that \eqref{eq:LCS_NN_cont} is not in the A-LCS form. Using Lemma~\ref{LCP_ReLU_equivalency}, we can can transform~\eqref{eq:LCS_NN_cont} into an A-LCS in a higher dimensional space. To see this, 
%
%
observe that \eqref{eq:LCS_NN_cont} is equivalent to
\begin{equation*}
	\begin{aligned}
		& x_{k+1} = A x_k + B (\bar{D} \bar{\lambda}_k+ \bar{z}) +\tilde{D} \tilde{\lambda}_k + \tilde{z},\\
		& \tilde{y}_k = \tilde{E} x_k +  \tilde{F} \tilde{\lambda}_k + H (\bar{D} \bar{\lambda}_k + \bar{z} )+ \tilde{c}, \\
		& \bar{y}_k = \bar{E} x_k + \bar{F} \bar{\lambda}_k + \bar{c}, \\
		& 0 \leq \tilde{\lambda}_k \perp \tilde{y}_k \geq 0, \\
		& 0 \leq \bar{\lambda}_k \perp \bar{y}_k \geq 0,
	\end{aligned}
\end{equation*}
after direct application of Lemma~\ref{LCP_ReLU_equivalency} where $\bar{\lambda} \in \mathbb{R}^{n_t}$.
We can write it succinctly as
\begin{equation}
	\label{eq:LCS_Controller}
	\begin{aligned}
		& x_{k+1} = A x_k + D \lambda_k + z, \\
		& y_k = E x_k + F \lambda_k + c, \\
		& 0 \leq \lambda_k \perp y_k \geq 0,
	\end{aligned}
\end{equation}
where $\lambda_k = \begin{bmatrix} \tilde{\lambda}_k \\ \bar{\lambda}_k \end{bmatrix}$, $y_k = \begin{bmatrix} \tilde{y}_k \\ \bar{y}_k \end{bmatrix}$, $D = \begin{bmatrix} \tilde{D} & B \bar{D} \end{bmatrix}$, $E = \begin{bmatrix} \tilde{E} \\ \bar{E} \end{bmatrix}$, $F = \begin{bmatrix}
	\tilde{F} & H \bar{D} \\ 0 & \bar{F} \end{bmatrix}$, $c = \begin{bmatrix} \tilde{c} + H \bar{z} \\ \bar{c} \end{bmatrix}$, and $z = B \bar{z} + \tilde{z}$. Here, the size of $x_k \in \mathbb{R}^{n_x}$ does not change, but notice that now $\lambda_k \in \mathbb{R}^{n_{\lambda}}$ where $n_{\lambda} = n_t + n_{ \tilde{\lambda} }$.
Using controllers of the form $\eqref{eq:NN_ReLU}$, we will exclusively consider the linear complementarity system model \eqref{eq:LCS_Controller} for notational compactness.

Similarly, one can consider the scenario where both the system dynamics and the controller are represented by ReLU neural networks as in $x_{k+1} = \phi_1(x_{k}) + B \phi_2(x_k)$, where $\phi_1$ represents the autonomous part of the dynamics (obtained by, for example, system identification) and $\phi_2$ is the controller. Using Lemma~\ref{LCP_ReLU_equivalency}, this system has an equivalent A-LCS representation similar to \eqref{eq:LCS_Controller}, but the details are omitted for brevity.

After obtaining an A-LCS representation of the closed-loop system, one can directly use the existing tools for stability analysis of complementarity systems, such as Lyapunov functions \cite{camlibel2007lyapunov} and semidefinite programming \cite{aydinoglu2020contact}. We will elaborate on this in the next section.

\section{Stability Analysis of the Closed-Loop System}

In this section, we provide sufficient conditions for stability in the sense of Lyapunov for an A-LCS.
Then, we show that the stability verification problem is equivalent to finding a feasible solution to a set of linear matrix inequalities (LMI's). To begin, consider the following Lyapunov function candidate that was introduced in \cite{camlibel2007lyapunov}:
\begin{equation}
	\label{eq:lyapunov_function}
	\begin{aligned}
		V(x_k, \lambda_k) 
		&= \begin{bmatrix}
			x_k \\ \lambda_k \\ 1
		\end{bmatrix}^\top \underbrace{\begin{bmatrix}
				P & Q & h_1 \\ Q^\top  & R & h_2 \\ h_1^T & h_2^T & h_3
		\end{bmatrix}}_{:=M}\begin{bmatrix}
			x_k \\ \lambda_k \\ 1
		\end{bmatrix},
	\end{aligned}
\end{equation}
where $P \in \mathbb{S}^{n_x}$, $Q \in \mathbb{R}^{n_x \times n_\lambda}$, $R \in \mathbb{S}^{n_\lambda}$, $h_1 \in \mathbb{R}^{n_x}$, $h_2 \in \mathbb{R}^{n_\lambda}$, and $h_3 \in \mathbb{R}$ are to be chosen.
Note that if $F$ in \eqref{eq:LCS_Controller} is a P-matrix, then $\lambda_k$ is a piecewise affine function of $x_k$, implying that the Lyapunov function \eqref{eq:lyapunov_function} is quadratic in the pair $(x_k, \lambda_k)$ but  it is \emph{piecewise} quadratic (PWQ) in the state $x_k$. If $F$ is not a P-matrix, then $V$ can be set valued since there can be multiple $\lambda_k$'s corresponding to each $x_k$. In either case, $V$ reduces to a common quadratic Lyapunov function in the special case $Q=R=0$. Therefore, \eqref{eq:lyapunov_function} is more expressive than a common quadratic Lyapunov function.

In the following theorem, we construct sufficient conditions for the stability of~\eqref{eq:LCS_Controller}, using the Lyapunov function \eqref{eq:lyapunov_function}. This is the discrete time version of the results in \cite{camlibel2007lyapunov}.
\begin{theorem}
	\label{stability_theorem}
	Consider the A-LCS in \eqref{eq:LCS_Controller} with the equilibrium $x_e = 0$, the Lyapunov function \eqref{eq:lyapunov_function} and a domain $\mathcal{X} \subseteq \mathbb{R}^n$. 
	If there exist $M \in \mathbb{S}^{n_x+n_\lambda+1}$, $\alpha_1 > 0$, and $\alpha_2> \alpha_3 \geq 0$ 
	such that 
	\begin{align*}
		& \alpha_1 ||x_k||_2^2	\leq V(x_k, \lambda_k) \leq \alpha_2 ||x_k||_2^2, \; \forall (x_k, \lambda_k) \in \Gamma_1, \\
		& V(x_{k+1}, \lambda_{k+1} ) - V(x_k, \lambda_k) \leq -\alpha_3 ||x_k||_2^2, \; \forall (x_k, \lambda_k, \lambda_{k+1}) \in \Gamma_2,
	\end{align*}
	where $x_{k+1} = A x_k + D \lambda_k + z$ and
	\begin{align*}
		&\Gamma_1 = \{ (x_k, \lambda_k) : \; 0 \leq \lambda_k \perp E x_k + F \lambda_k + c \geq 0, \ x_k \in \mathcal{X} \}, \\
		&\Gamma_2 = \{ (x_k, \lambda_k, \lambda_{k+1}) : 0 \leq \lambda_k \perp E x_k + F \lambda_k + c \geq 0,\\
		& \qquad \qquad \quad 0 \leq \lambda_{k+1} \perp  E x_{k+1} + F \lambda_{k+1} + c \geq 0, \ x_k \in \mathcal{X}    \}.
	\end{align*}
	Then the equilibrium is Lyapunov stable if $\alpha_3 = 0$ and geometrically stable if $\alpha_2 > \alpha_3 > 0$.
\end{theorem}
\begin{proof}
	Observe that for all $(x_k, \lambda_k)$:
	\begin{equation*}
		\alpha_1 ||x_k||_2^2 \leq V(x_k, \lambda_k) \leq V(x_0, \lambda_0) \leq \alpha_2 ||x_0||_2^2,
	\end{equation*}
	and Lyapunov stability follows. For geometric stability, notice that Lyapunov decrease condition is equivalent to $V(x_{k+1}, \lambda_{k+1}) - \gamma V(x_k, \lambda_k) \leq 0$, 
	for some $\gamma \in (0,1)$. Then
	\begin{equation*}
		\alpha_1 ||x_k||_2^2 \leq V(x_k, \lambda_k) \leq \gamma^k V(x_0, \lambda_0) \leq \alpha_2 \gamma^k ||x_0||_2^2.
	\end{equation*}
	The result follows.
\end{proof}

Note that we do not require $M$ in \eqref{eq:lyapunov_function} to be positive definite to satisfy the requirements of Theorem \ref{stability_theorem}.
In light of this theorem, we must solve the following feasibility problem to verify that if the equilibrium of the  closed-loop system \eqref{eq:LCS_Controller} is stable on $\mathcal{X}$:
\begin{alignat}{2}
	\label{eq:feasability_LCS}
	& \underset{}{\text{find}} && P, Q, R, h_1, h_2, h_3, \alpha_1, \alpha_2, \alpha_3 \\
	\notag& \text{s.t.}  \quad && \alpha_1 ||x_k||_2^2	\leq V(x_k, \lambda_k) \leq \alpha_2 ||x_k||_2^2, \; \text{for} \; (x_k, \lambda_k) \in \Gamma_1,\\
	\notag& && \Delta V \leq -\alpha_3 ||x_k||_2^2, \; \text{for} \; (x_k, \lambda_k, \lambda_{k+1}) \in \Gamma_2,
\end{alignat}
where $\Delta V = V(x_{k+1}, \lambda_{k+1}) - V(x_k, \lambda_k)$.
In the following proposition, we turn \eqref{eq:feasability_LCS} with $\mathcal{X} = \mathbb{R}^n$ into an LMI feasibility problem using the S-procedure \cite{boyd1994linear}.

\begin{proposition} The following LMI's solve \eqref{eq:feasability_LCS} with $\mathcal{X} = \mathbb{R}^n$:
	\begin{subequations} \label{thm: LMIs}
		\begin{align}
			& T_1 - S_1^T W_1 S_1 - \frac{1}{2}(S_{3,1} + S_{3,1}^\top) \succeq 0, \label{thm: LMIs 1}\\
			& T_2 + S_1^\top W_2 S_1 + \frac{1}{2}(S_{3,2} + S_{3,2}^\top) \preceq 0, \label{thm: LMIs 2}\\
			& T_3 + S_2^T W_3 S_2 + S_5^T W_4 S_5 + \frac{1}{2} [ (S_4 + S_4^\top) + (S_6 + S_6^\top) ]  \succeq 0, \label{thm: LMIs 3}
		\end{align}
	\end{subequations}
	where $G_1 = D^T P z + D^T h_1 - h_2$, $G_2 = z^T P D + h_1^T D - h_2$, $G_3 = z^T Q + h_2^T$, $S_1 = \begin{bmatrix}
		E & F & c \\ 0 & I & 0 \\ 0 & 0 & 1
	\end{bmatrix}$, $S_2 = \begin{bmatrix}
		E & F & 0 & c \\ 0 & I & 0 & 0 \\ 0 & 0 & 0 & 1
	\end{bmatrix}$, $S_{3,i} = \begin{bmatrix}
		0 & 0 & 0 \\ J_i E & J_i F &  J_i c \\ 0 & 0 & 0
	\end{bmatrix}$, 
	$\quad S_4 = \begin{bmatrix}
		0 & 0 & 0 & 0 \\ J_3 E & J_3 F & 0 & J_3 c \\ 0 & 0 & 0 & 0
	\end{bmatrix}$, $S_5 = \begin{bmatrix}
		E A & E D & Fc & Ec z + c \\ 0 & 0 & I & 0 \\ 0 & 0 & 0 & 1
	\end{bmatrix}$, 
	
	$S_6 = \begin{bmatrix}
		0 & 0 & 0 & 0 \\ J_4 E A & J_4 E D & J_4 F & J_4 E z + c \\ 0 & 0 & 0 & 0
	\end{bmatrix}$, $T_1 = \begin{bmatrix}
		P - \alpha_1 I & Q & h_1 \\ Q^\top  & R & h_2 \\ h_1^T & h_2^T & h_3
	\end{bmatrix}$, 
	$ T_2 = \begin{bmatrix}
		P - \alpha_2 I & Q & h_1 \\ Q^\top  & R & h_2 \\ h_1^T & h_2^T & h_3
	\end{bmatrix}$,
	\begin{equation*}
		T_3 = - \begin{bmatrix}
			A^T P A - P + \alpha_3 I & A^T P D - Q & A^T Q & A^T P z - h_1 \\
			D^T P A - Q^T & D^T P D - R & D^T Q & G_1 \\
			Q^T A & Q^T D & R & Q^T z + h_2 \\
			z^T P A - h_1^T & G_2 & G_3 & z^T P z - h_1^T z
		\end{bmatrix}.
	\end{equation*}
	Here, $W_i$ are decision variables with non-negative entries, and $ J_i= \operatorname{diag}(\tau_i)$ where $\tau_i \in \mathbb{R}^m$ are free decision variables. 
\end{proposition}

\begin{proof}
	First define $e_k^\top = \begin{bmatrix} x_k^\top  & \lambda_k^\top & 1 \end{bmatrix}$. By left and right multiplying both sides of \eqref{thm: LMIs 1} by $e_k^\top$ and $e_k$, respectively, we obtain
	\begin{align*}
		V(x_k,\lambda_k) - \alpha_1 \|x_k\|_2^2 \geq \begin{bmatrix}
			y_k \\ \lambda_k \\ 1
		\end{bmatrix}^\top W_1 \begin{bmatrix}
			y_k \\ \lambda_k \\ 1 \end{bmatrix} \!+ \! 2 \lambda_k^\top \operatorname{diag}(\tau_1) y_k
	\end{align*}
	The right hand side is non-negative due to the complementarity constraint $0 \leq \lambda_k \perp y_k \geq 0$. Similarly, by left and right multiplying both sides of \eqref{thm: LMIs 2} by $e_k^\top$ and $e_k$, respectively, we obtain
	\begin{align*}
		\alpha_2 \|x_k\|_2^2 - V(x_k,\lambda_k) \geq \begin{bmatrix}
			y_k \\ \lambda_k \\ 1
		\end{bmatrix}^\top W_2 \begin{bmatrix}
			y_k \\ \lambda_k \\ 1 \end{bmatrix} \!+ \! 2 \lambda_k^\top \operatorname{diag}(\tau_2) y_k
	\end{align*}
	Again, the right hand side is non-negative due to the complementarity constraint $0 \leq \lambda_k \perp y_k \geq 0$. 
	
	Now, we define $p_k^\top = \begin{bmatrix} x_k^\top  & \lambda_k^\top & \lambda_{k+1}^\top & 1 \end{bmatrix}$. Notice that if we left and right multiply both sides of \eqref{thm: LMIs 3} by $p_k^\top$ and $p_k$, we obtain
	\begin{align*}
		-\Delta V - \alpha_3 ||x_k||_2^2 \geq & \begin{bmatrix}
			y_k \\ \lambda_k \\ 1 \end{bmatrix}^\top W_3 \begin{bmatrix}
			y_k \\ \lambda_k \\ 1 \end{bmatrix} + \begin{bmatrix}
			y_{k+1} \\ \lambda_{k+1} \\ 1 \end{bmatrix}^\top W_4 \begin{bmatrix}
			y_k \\ \lambda_k \\ 1 \end{bmatrix} \\ 
		& + \! 2 \lambda_k^\top \operatorname{diag}(\tau_3) y_k + \! 2 \lambda_{k+1}^\top \operatorname{diag}(\tau_4) y_{k+1}
	\end{align*}
	Similarly, all the terms on the right hand side are non-negative since $0 \leq \lambda_k \perp y_k \geq 0$ for all $k$. This concludes the proof.
\end{proof}

Notice that \eqref{eq:feasability_LCS} captures the non-smooth structure of the LCS combined with the ReLU neural network controller. In addition to that, we can assign a different quadratic function to each polyhedral partition that is created by the neural network without enumerating those partitions by exploiting the complementarity structure of the neural network. Observe that \eqref{thm: LMIs 1}, \eqref{thm: LMIs 2} are LMI's of size $(n_x+n_\lambda+1)$, and \eqref{thm: LMIs 3} is an LMI of size $(n_x+2n_\lambda+1)$. 

Note that Theorem \ref{thm: LMIs} is a global result for $\mathcal{X} = \mathbb{R}^n$. We can adapt the theorem to bounded regions $\mathcal{X}$ containing the origin.



\begin{remark}
	For the equilibrium $x_e=0$, the region of attraction is defined as
	\begin{equation*}
		\mathcal{R} = \{ x_0 : \lim_{k \rightarrow \infty} ||x_k|| = 0 \} .
	\end{equation*}
	If one adds (to the left side) $-\eta_1 L_1$ to \eqref{thm: LMIs 1}, $+\eta_2 L_1$ to \eqref{thm: LMIs 2} and $+\eta_3 L_2$ to \eqref{thm: LMIs 3} where
	\begin{align*}
		L_1 = \begin{bmatrix}
			-P & -Q & -h_1 \\ -Q^\top & -R & -h_2 \\ -h_1^\top & -h_2^\top & \xi -h_3
		\end{bmatrix}, L_2 =  \begin{bmatrix}
			-P & -Q & 0 & h_1\\ -Q^T & R & 0 & -h_2\\ 0 & 0 & 0 & 0 \\ -h_1^T & -h_2^T & 0 & \xi - h_3
		\end{bmatrix},
	\end{align*}
	and $\eta_i$ are non-negative scalar variables, then the closed-loop system is geometrically stable for $\alpha_2>\alpha_3>0$ and the sub-level set
	\begin{equation*}
		\mathcal{V}_{\xi} = \{x : V(x,\lambda) \leq \xi \ \forall (x,\lambda) \in \Gamma_1 \},
	\end{equation*}
	is an approximation of the ROA, i.e., $\mathcal{V}_\xi \subseteq \mathcal{R}$. To see this, note that the resulting matrix inequality would imply
	\begin{align*}
		&\alpha_1 \|x_k\|_2^2 + \eta_1 (\xi- V(x_k, \lambda_k)) )\leq V(x_k,\lambda_k) \leq \eta_2 (V(x_k, \lambda_k)-\xi) + \alpha_2 \|x_k\|_2^2 \\
		&V(x_{k+1}, \lambda_{k+1}) - V(x_k, \lambda_k) + \alpha_3 \|x_k\|_2^2+ \eta_3 (\xi- V(x_k, \lambda_k) ) \leq 0.
	\end{align*}
	From the first inequality, if $ V(x_k, \lambda_k)  \leq \xi$, then $	\alpha_1 \|x_k\|_2^2 \leq V(x_k,\lambda_k) \leq \alpha_2 \|x_k \|_2^2$. 
	From the second inequality, for some $\gamma \in (0,1)$ we have $	V(x_{k+1},\lambda_{k+1}) \leq \gamma V(x_k,\lambda_k) \leq \xi$. By induction, if $V(x_0,\gamma_0) \leq \xi$, then $\alpha_1 \|x_k\|_2^2 \leq V(x_k,\gamma_k) \leq \gamma^k V(x_0,\gamma_0) \leq \gamma^k \alpha_2 \|x_0\|_2^2$.
\end{remark}

\begin{remark}
	In order to prove the Lyapunov conditions over the ellipsoid $\mathcal{X} = \{ x: x^T N x \leq \xi \}$, one can add (to the left side) $-\beta_1 N_1$ to \eqref{thm: LMIs 1}, $+\beta_2 N_1$ to \eqref{thm: LMIs 2} and $+\beta_3 N_2$ to \eqref{thm: LMIs 3} where
	\begin{align*}
		N_1 = \begin{bmatrix}
			-N & 0 & 0 \\ 0 & 0 & 0 \\ 0 & 0 & \xi
		\end{bmatrix}, N_2 = \begin{bmatrix}
			-N & 0 & 0 & 0\\ 0 & 0 & 0 & 0\\ 0 & 0 & 0 & \xi
		\end{bmatrix},
	\end{align*}
	and $\beta_i$ are non-negative scalar variables. 	
\end{remark}

\begin{figure}[t]
	\centering
	\includegraphics[width=0.5\linewidth]{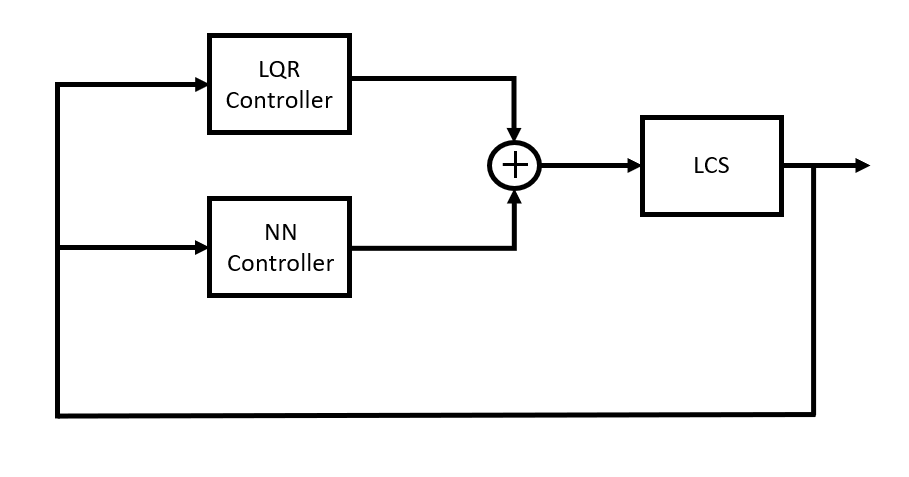}
	\caption{Block diagram of the closed-loop system.}
	\label{fig:cont_arch}
\end{figure}

\section{Examples}

We use YALMIP \cite{lofberg2004yalmip} toolbox with MOSEK \cite{mosek2010mosek} to formulate and solve the linear matrix inequality feasibility problems. PATH \cite{dirkse1995path} has been used to solve the linear complementarity problems when simulating the system. PyTorch is used for training neural network controllers \cite{paszke2017automatic}. The experiments are done on a desktop computer with the processor Intel \emph{i7-4790} and \emph{8GB RAM} unless stated otherwise. For all of the experiments, we consider the closed-loop system in Figure \ref{fig:cont_arch} and the linear-quadratic regulator controller is designed with state penalty matrix $Q^\text{LQR} = 10 I$ and input penalty matrix $R^\text{LQR} = I$ unless stated otherwise.

\subsection{Double Integrator}

In this example, we consider a double integrator model:
\begin{equation*}
	\begin{aligned}
		& x_{k+1} = A x_k  + B u_k,
	\end{aligned}
\end{equation*}
where $\tilde{A} = \begin{bmatrix}
	1 & 1 \\ 0 & 1
\end{bmatrix}$, $B = \begin{bmatrix}
	0.5 \\ 1
\end{bmatrix}$, and $A = \tilde{A} + B K_\text{LQR}$, where LQR gains are $Q^\text{LQR} = 0.1I$ and $R^\text{LQR} = 1$. This simple model serves as an example where we approximate an explicit model predictive controller (explicit MPC) \cite{bemporad2002explicit} using a neural network and verify the stability of the resulting system. We consider the state and input constraints:
\begin{figure}[t]
	\centering
	\begin{subfigure}[t]{.5\textwidth}
	\centering
	\includegraphics[width=1\linewidth]{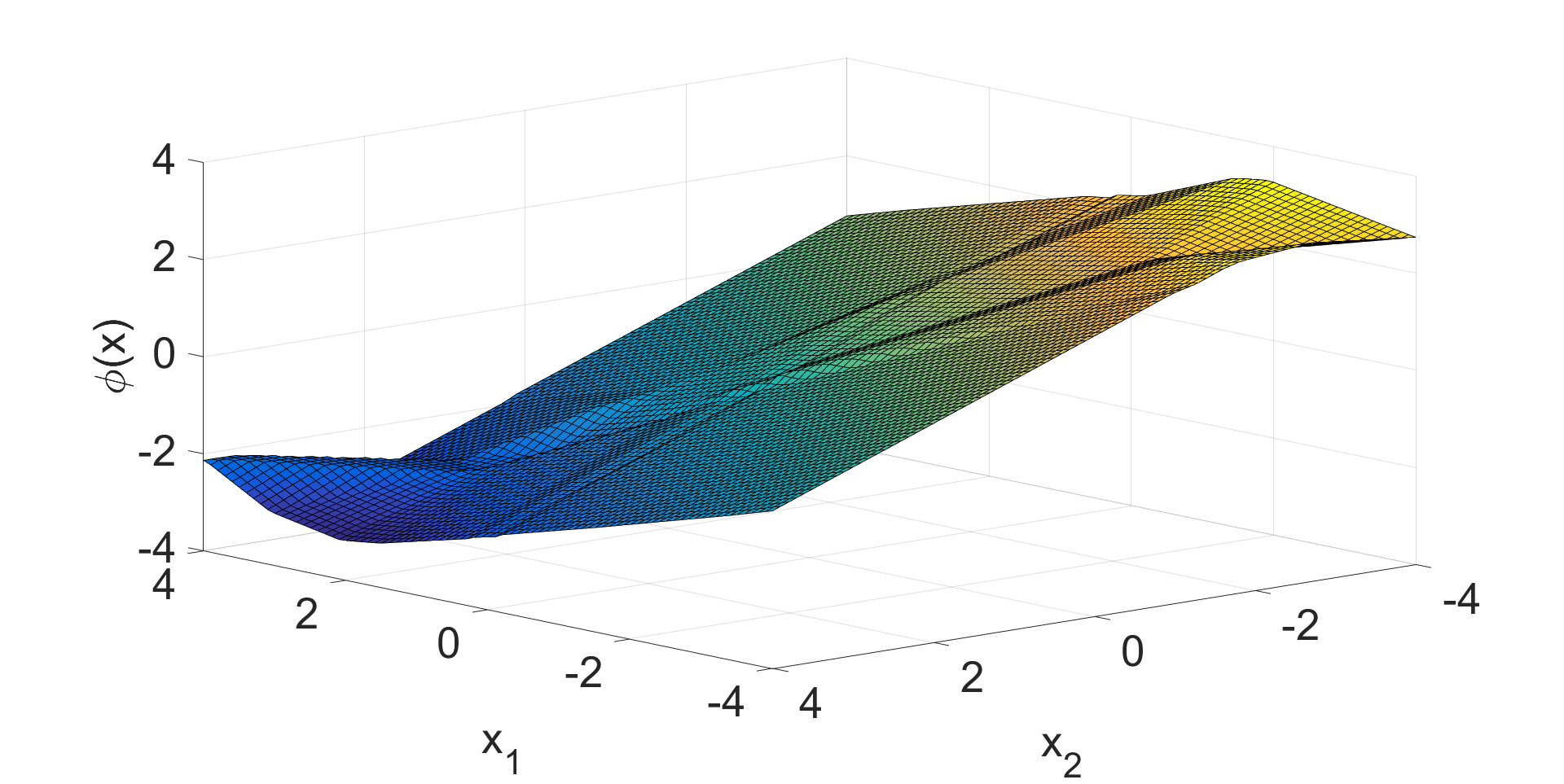}
	\caption{Neural network ($\phi$) policy.}
	\label{fig:policy}
	\end{subfigure}%
	\begin{subfigure}[t]{.5 \textwidth}
		\centering
		\includegraphics[width=1\linewidth]{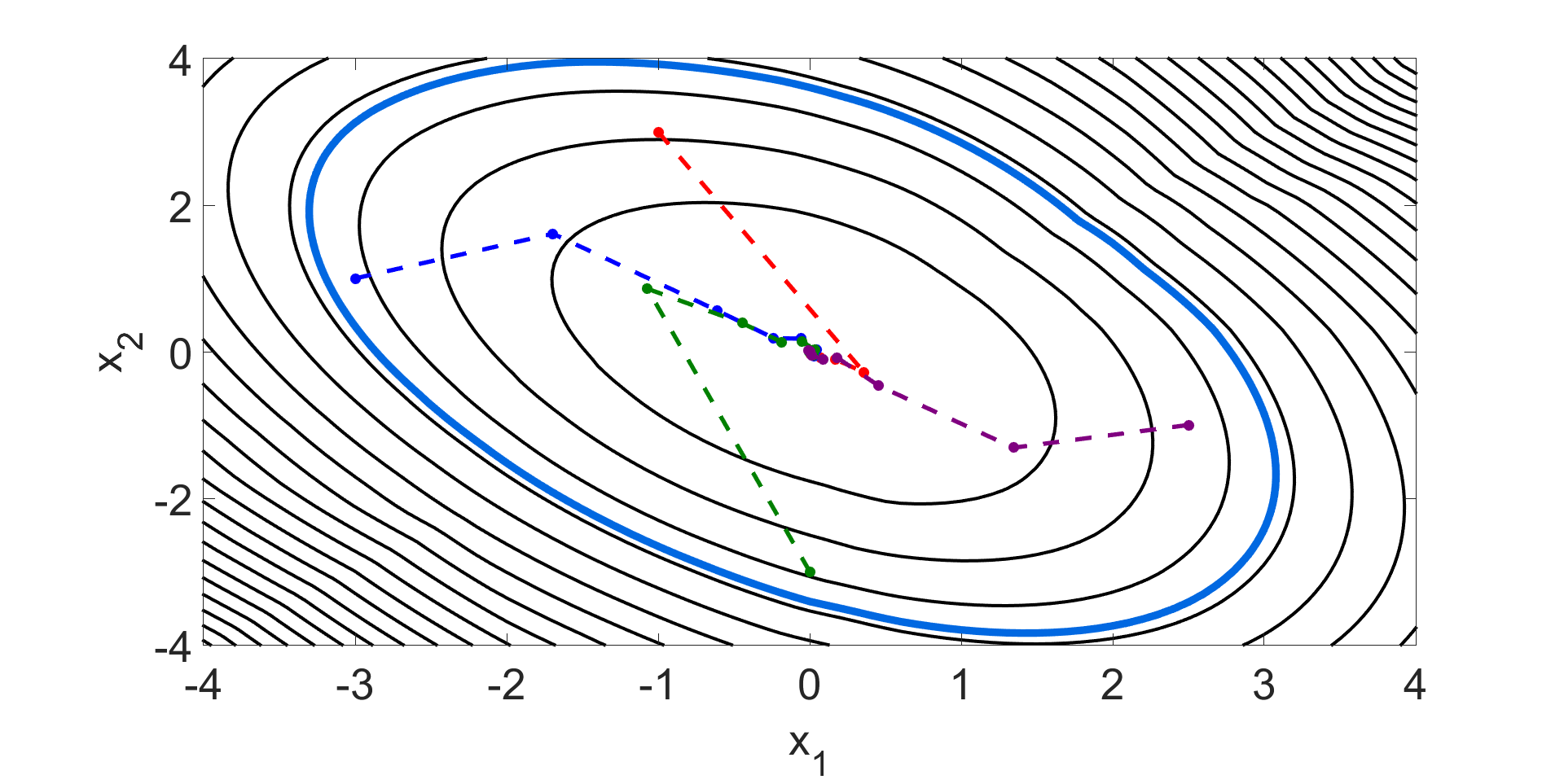}
		\caption{Sublevel sets with four different trajectories. One sublevel set that lies in the constraint set is shown in blue.}
		\label{fig:int_levelsets}
	\end{subfigure}
	\begin{subfigure}{.7 \textwidth}
	\centering
	\includegraphics[width=1\linewidth]{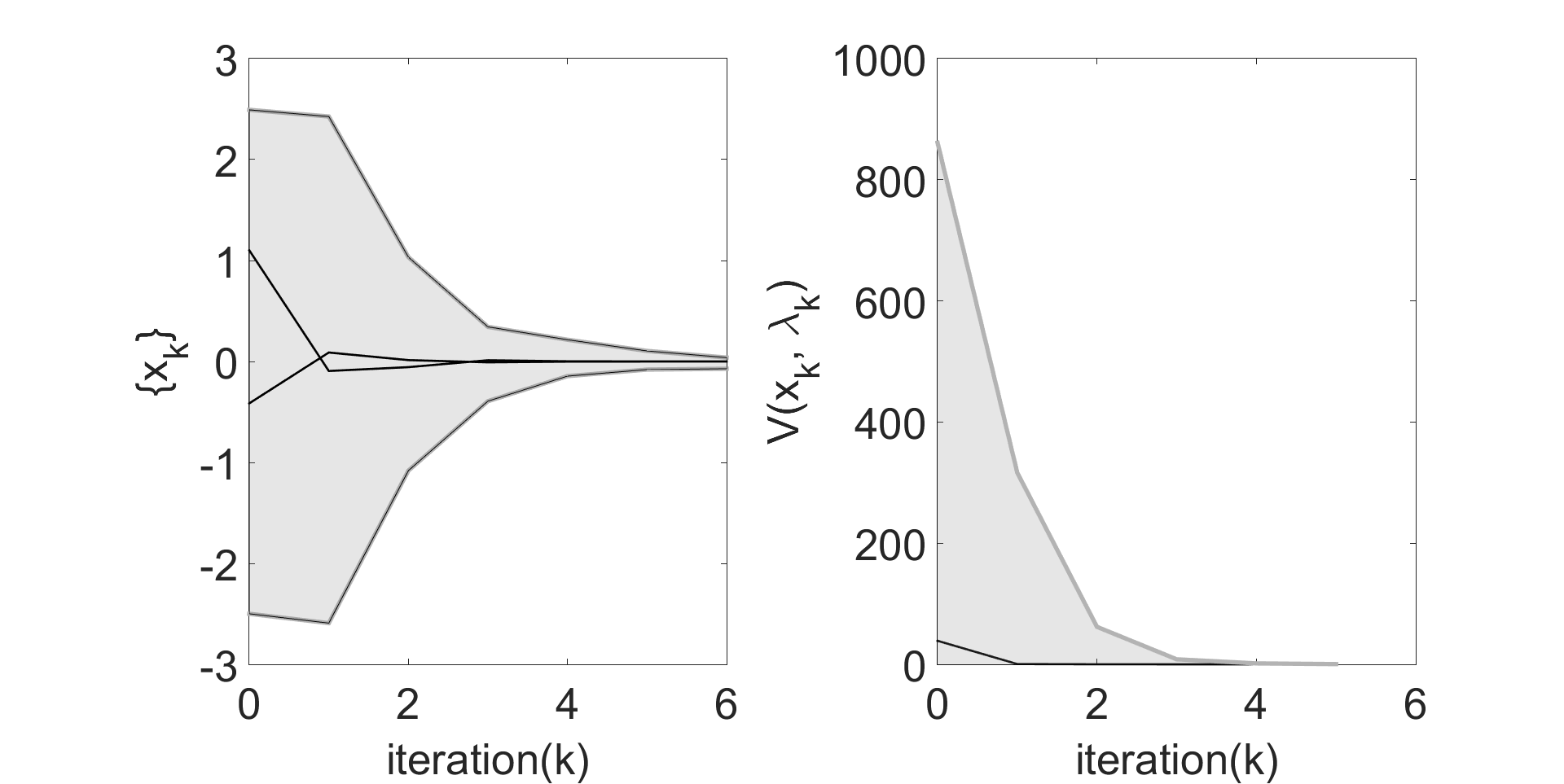}
	\caption{Envelopes for 1000 trajectories and their corresponding Lyapunov functions (in gray) with a sample trajectory (in black).}
	\label{fig:lyap_double_int}
	\end{subfigure}
	\caption{Experimental results for the double integrator example.}
\end{figure}

\begin{equation*}
	\mathcal{X} = \{ x : \begin{bmatrix}
		-4 \\ -4
	\end{bmatrix} \leq x \leq \begin{bmatrix}
		4 \\ 4
	\end{bmatrix}  \}, \; \mathcal{U} = \{ u : -3 \leq u \leq 3 \},
\end{equation*}
and obtain 2000 samples of the form $(x, \pi_\text{MPC}(x))$ with $N=10$, $Q^\text{MPC}=10I$, and $R^\text{MPC}=1$. Next we approximate the explicit MPC controller using a ReLU network $\phi(x)$ with two layers and 10 neurons in each layer as in Figure \ref{fig:policy}. Now, consider the closed-loop system:
\begin{equation}
	\label{eq:double_int_linear}
	x_{k+1} = A x_k + B \phi(x_k).
\end{equation}
First, we find the equivalent LCP representation of $\phi(x)$ using Lemma \ref{LCP_ReLU_equivalency}. Then, we write the equivalent LCS representation of \eqref{eq:double_int_linear} as described in Section \ref{sub:complementaritysys_nncont}.
We computed the piece-wise quadratic Lyapunov function of the form \eqref{eq:lyapunov_function} and verified exponential stability in 1.1 seconds. The sublevel sets of the Lyapunov functions are plotted in Figure \ref{fig:int_levelsets}. We also present the envelopes of 1000 trajectories with their corresponding Lyapunov functions in Figure \ref{fig:lyap_double_int}.

\subsection{Cart-pole with Soft Walls}

We consider the regulation problem of a cart-pole with soft-walls as in Figure \ref{fig:cartpole}. This problem has been studied in \cite{marcucci2020warm, deits2019lvis, aydinoglu2020stabilization} and is a benchmark in contact-based control algorithms. In this model, $x_1$ represents the position of the cart, $x_2$ represents the angle of the pole and $x_3$, $x_4$ are their time derivatives respectively. Here, $\lambda_1$ and $\lambda_2$ represent the contact force applied by the soft walls to the pole from the right and left walls, respectively. We consider the linearized model around $x_2 = 0$:
\begin{figure}[t]
	\centering
	\begin{subfigure}[t]{.5\textwidth}
		\centering
		\includegraphics[width=1\linewidth]{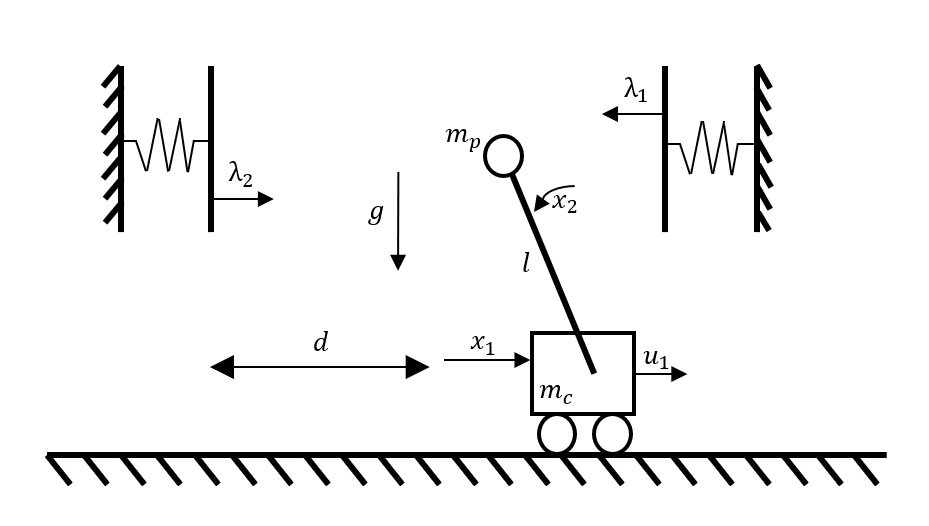}
		\caption{The cart-pole system.}
		\label{fig:cartpole}
	\end{subfigure}%
	\begin{subfigure}[t]{.5 \textwidth}
		\centering
		\includegraphics[width=1\linewidth]{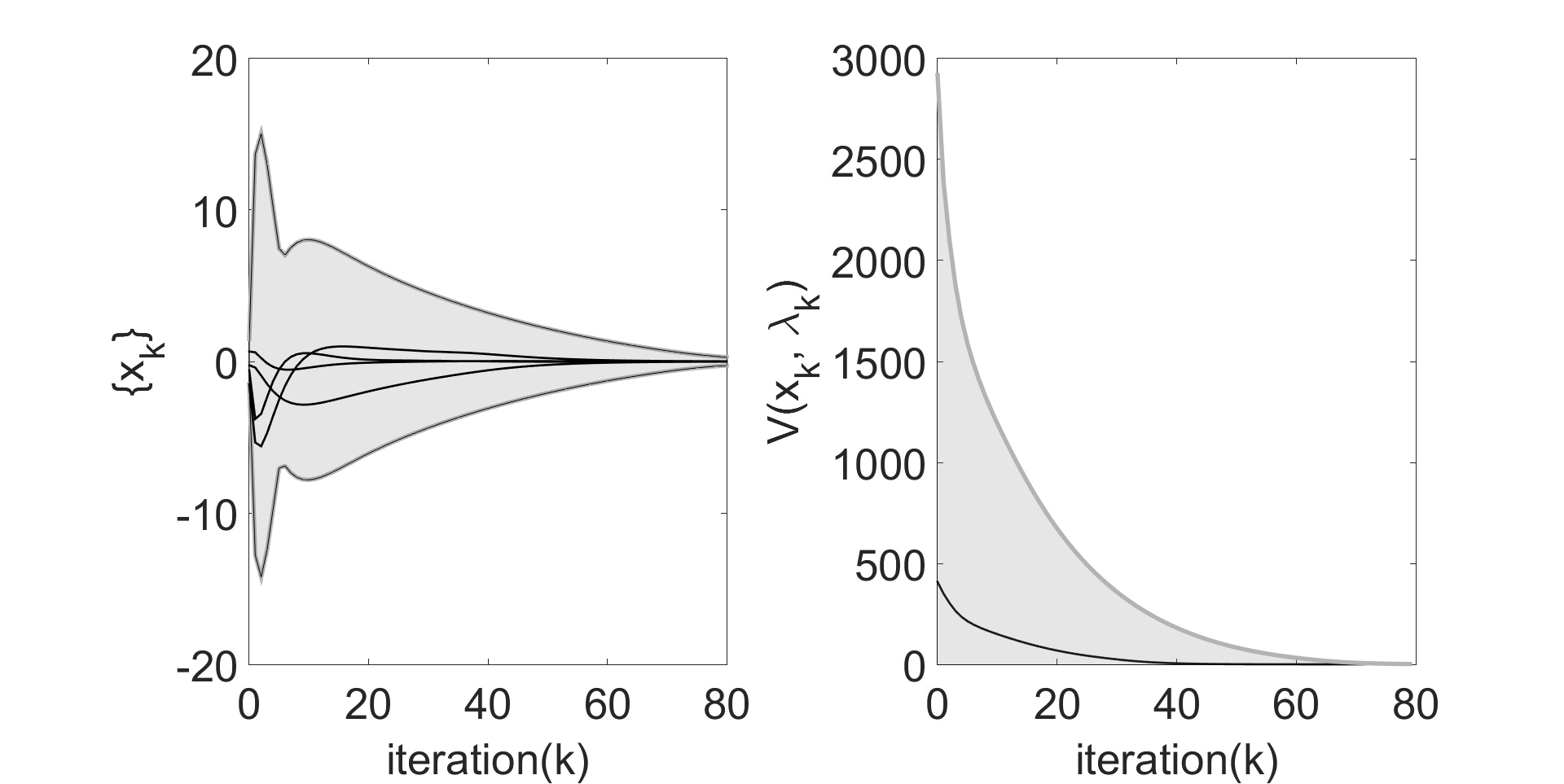}
		\caption{Envelopes for 1000 trajectories and the corresponding Lyapunov functions (in gray) with a sample trajectory (in black).}
		\label{fig:lyap_cartpole}
	\end{subfigure}
	\caption{Experimental results for the cart-pole example.}
\end{figure}
\begin{align*}
	& \dot{x}_1 = x_3, \\
	& \dot{x}_2 = x_4, \\
	& \dot{x}_3 = g \frac{m_p}{m_c} x_2 + \frac{1}{m_c} u_1, \\
	& \dot{x}_4 = \frac{g (m_c + m_p)}{l m_c} x_2 + \frac{1}{l m_c} u_1 + \frac{1}{l m_p} \lambda_1 - \frac{1}{l m_p} \lambda_2, \\
	& 0 \leq \lambda_1 \perp l x_2 - x_1 + \frac{1}{k_1} \lambda_1 + d \geq 0, \\
	& 0 \leq \lambda_2 \perp x_1 - l x_2  + \frac{1}{k_2} \lambda_2 + d \geq 0,
\end{align*}
where $m_c=1$ is the mass of the cart, $m_p = 1$ is the mass of the pole, $l=1$ is the length of the pole, $k_1 = k_2 = 1$ are the stiffness parameter of the walls, $d=1$ is the distance between the origin and the soft walls. Then, we discretize the dynamics using the explicit Euler method with time step $T_s = 0.1$ to obtain the system matrices:

$\tilde{A} = \begin{bmatrix}
	1 & 0 & 0.1 & 0 \\
	0 & 1 & 0 & 0.1 \\
	0 & 0.981 & 1 & 0 \\
	0 & 1.962 & 0 & 1
\end{bmatrix}$, $B = \begin{bmatrix}
	0 \\ 0 \\ 0.1 \\ 0.1
\end{bmatrix}$, $\tilde{D} = \begin{bmatrix}
	0 & 0 \\
	0 & 0 \\
	0 & 0 \\
	-0.1 & 0.1 
\end{bmatrix}$, $\tilde{E} = \begin{bmatrix}
	-1 & 1 & 0 & 0 \\ 1 & -1 & 0 & 0
\end{bmatrix}$, $\tilde{F} = \begin{bmatrix}
	1 & 0 \\ 0 & 1
\end{bmatrix}$, $\tilde{c} = \begin{bmatrix}
	1 \\ 1
\end{bmatrix}$, $A = \tilde{A} + B K_{LQR}$
and, $K_{LQR}$ is the gain of the linear-quadratic regulator that stabilizes the linear system $(\tilde{A},B)$. However, the equilibrium $x_e = 0$ is not globally stable due to the soft walls.

\begin{figure}[t]
	\centering
	\begin{subfigure}[t]{.5\textwidth}
		\centering
		\includegraphics[width=0.7\linewidth]{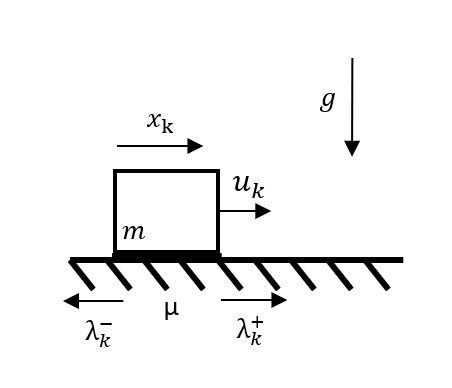}
		\caption{A box on a surface with friction.}
		\label{fig:box}
	\end{subfigure}%
	\begin{subfigure}[t]{.5 \textwidth}
		\centering
		\includegraphics[width=1\linewidth]{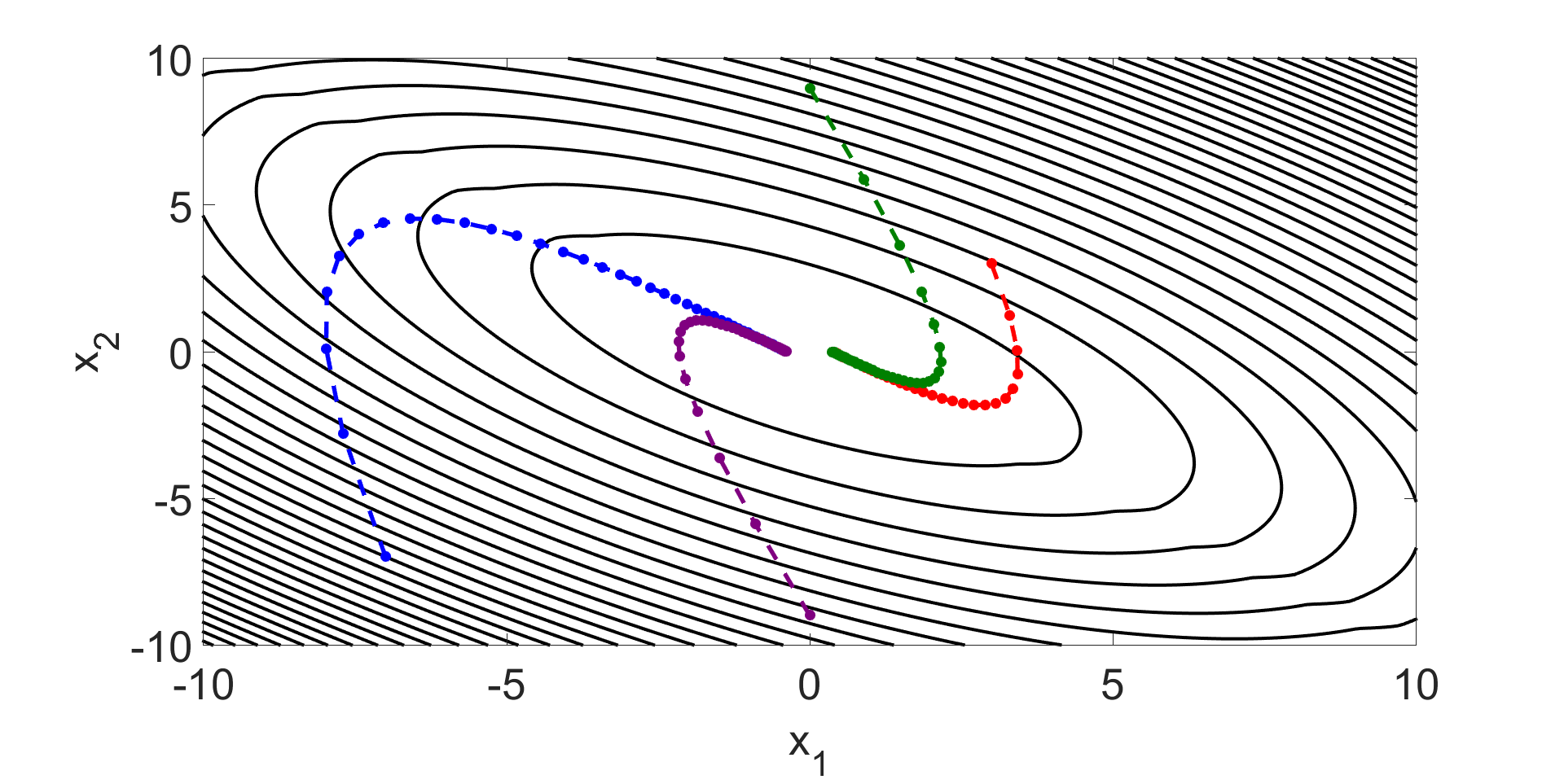}
		\caption{Sublevel sets of the piece-wise quadratic Lyapunov function $V(x_k, \lambda_k)$ with four different trajectories.}
		\label{fig:sublevel_fric}
	\end{subfigure}
	\begin{subfigure}{.7 \textwidth}
		\centering
		\includegraphics[width=1\linewidth]{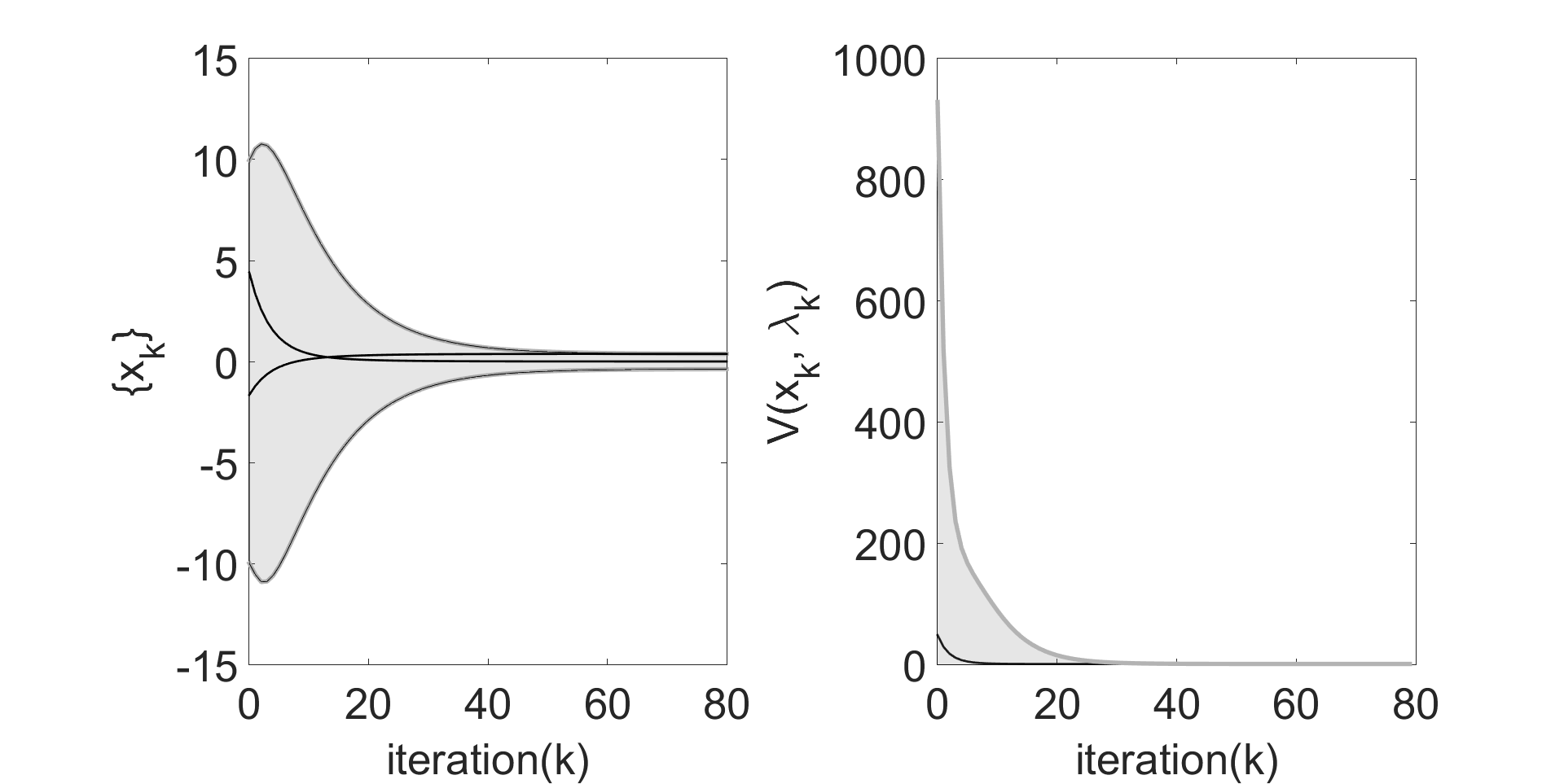}
		\caption{Envelopes for 1000 trajectories (in gray) with a sample trajectory (in black).}
		\label{fig:lyap_box}
	\end{subfigure}
	\caption{Experimental results for the box with friction example.}
\end{figure}

We solve the optimal control problem (Appendix \ref{ap:opt_cont}, \eqref{eq:opt_original}) with $N=10$, $Q^\text{OPT} = 10 I$, $R^\text{OPT}= 1$, and $Q^\text{OPT}_N$ as the solution of the discrete algebraic Riccati equation to generate samples of the form $(x, \pi_\text{OPT}(x))$. For this particular problem, we generate 4000 samples and we train a neural network $\phi(x)$ with two layers, each with 10 neurons, to approximate the optimal controller $\pi_\text{OPT}$. We used the ADAM optimizer to do the training. Then, we analyze the linear complementarity system with the neural network controller $u_k = \phi(x_k)$. Following the procedure in Section \ref{sec:comp_NN}, we first express the neural network as a linear complementarity problem using Lemma \ref{LCP_ReLU_equivalency} and then transform the LCS with the NN controller into the form \eqref{eq:LCS_Controller}. We compute a Lyapunov function of the form \eqref{eq:lyapunov_function} in 1.7 seconds that verifies that the closed-loop system with the neural network controller $\phi(x)$ is globally exponentially stable. For this example, a common Lyapunov function is enough to verify stability.
In Figure \ref{fig:lyap_cartpole}, we present the envelopes for 1000 trajectories.

\subsection{Box with Friction}

In this example, we consider the regulation task of a box on a surface as in Figure \ref{fig:box}. 
This simple model serves as an example where the contact forces $\lambda_k$ are not unique due to Coulomb friction between the surface and the box.
Here, $x_1$ is the position of the cart, $x_2$ is the velocity of the cart, $u$ is the input applied to the cart, $g=9.81$ is the gravitational acceleration, $m=1$ is the mass of the cart, and $\mu = 0.1$ is the coefficient of friction between the cart and the surface. 
The system can be modeled by:
\begin{equation}
	\label{eq:box_friction_equations}
	\begin{aligned}
		& x_{k+1} = A x_k + B u_k + \tilde{D} \tilde{\lambda}_k, \\	
		& 0 \leq \tilde{\lambda}_k \perp \tilde{E} x_k + \tilde{F} \tilde{\lambda}_k + \tilde{c} \geq 0,
	\end{aligned}
\end{equation}
where $\tilde{A} = \begin{bmatrix}
	1 & 0.1 \\ 0 & 1
\end{bmatrix}$, $B = \begin{bmatrix}
	0 \\ 0.1
\end{bmatrix}$, $\tilde{D} = \begin{bmatrix}
	0 & 0 & 0 \\ 0.1 & -0.1 & 0
\end{bmatrix}$, $\bar{E} = \begin{bmatrix}
	0 & 1 \\ 0 & -1 \\ 0 & 0
\end{bmatrix}$, $\tilde{F} = \begin{bmatrix}
	1 & -1 & 1 \\ -1 & 1 & 1 \\ -1 & -1 & 0
\end{bmatrix}$, $\tilde{c} = \begin{bmatrix}
	0 \\ 0 \\ 0.981
\end{bmatrix}$, $H = \begin{bmatrix}
	1 \\ -1 \\ 0 \end{bmatrix}$, $\tilde{E} = \bar{E} + H K_{LQR}$, $A = \tilde{A} + B K_{LQR}$
and, $K_{LQR}$ is the gain that (the linear-quadratic regulator controller) stabilizes the linear system $(\tilde{A},B)$. Observe that the matrix $\tilde{F}$ is not a P-matrix, hence for a given $x_k$, the contact forces $\lambda_k$ are not unique.
Similar to the previous example, we generate 2000 samples $(x, \pi_\text{OPT}(x))$ for the LCS in \eqref{eq:box_friction_equations} with $N=5$, $Q^\text{OPT} = Q^\text{OPT}_N = 0.1$, $R^\text{OPT}=1$ and train a neural network $\phi(x)$ that approximates the optimal controller. Then we convert the system in \eqref{eq:box_friction_equations} with the neural network controller $\phi(x)$ into the form \eqref{eq:LCS_Controller}. Next, we compute the piece-wise quadratic Lyapunov function (with sublevel sets shown in Figure \ref{fig:sublevel_fric}) of the form \eqref{eq:lyapunov_function} in 1.6 seconds such that the exponential stability condition is verified outside a ball around the origin, $\mathcal{D} = \{ x : ||x||_2^2 > 0.6 \}$. More precisely, we prove convergence to a set (smallest sublevel set of $V$ that contains $\mathcal{D}$) which contains the equilibrium. This is expected because the trajectories do not reach the origin due to stiction. We demonstrate the envelopes for 1000 trajectories and their respective Lyapunov functions in Figure \ref{fig:lyap_box}.

\begin{figure}[b]
	\centering
	\includegraphics[width=1\linewidth]{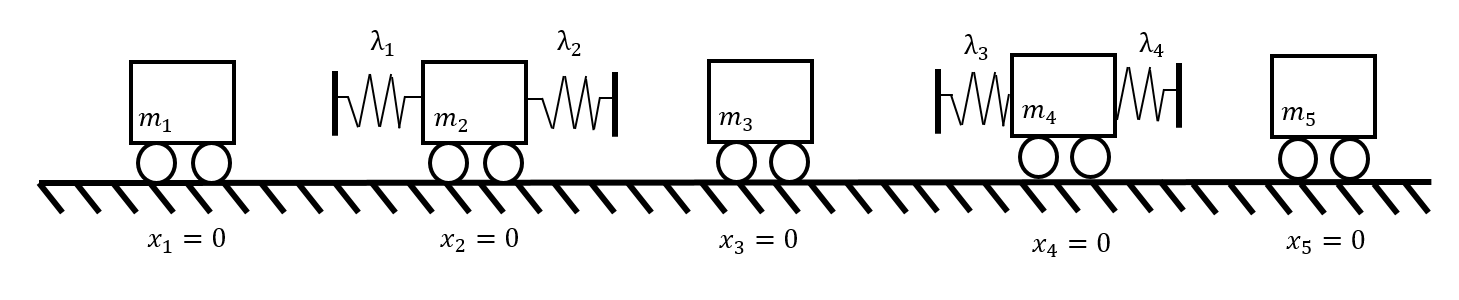}
	\caption{Regulation task of five carts to their respective origins.}
	\label{fig:five_carts}
\end{figure}

\subsection{Five Carts}

We consider the regulation task of five carts as in Figure \ref{fig:five_carts}. Here $x_i$ describes the state of the $i$-th cart, the interaction between the carts is modeled by soft springs represented by $\lambda_i$, and all carts can be controlled via the applied force $u_i$. We approximate Newtons’s second law with a force balance equation and obtain the following quasi-static model:
\begin{align*}
	& x_{k+1}^{(1)} = x_k^{(1)} + u_k^{(1)} - \lambda_k^{(1)}, \\
	& x_{k+1}^{(i)} = x_k^{(i)} + u_k^{(i)} - \lambda_k^{(i-1)} - \lambda_k^{(i)}, \; \text{for} \; i=2,3,4, \\
	& x_{k+1}^{(5)} = x_k^{(5)} + u_k^{(5)} + \lambda_k^{(4)}, \\
	& 0 \leq \lambda_k^{(i)} \perp x_k^{(i+1)} - x_k^{(i)} + \lambda_k^{(i)} \geq 0.
\end{align*}
We designed an LQR controller with with state penalty matrix $Q^\text{LQR} = I$ and input penalty matrix $R^\text{LQR} = I$. Then, we solve the optimal control problem (Appendix \ref{ap:opt_cont}, \eqref{eq:opt_original}) with $N=10$, $Q^\text{OPT} = Q^\text{OPT}_N = 10 I$, and $R^\text{OPT}= 1$ to generate 2000 samples of the form $(x, \pi_\text{OPT}(x))$.
Using these samples, we train $\phi(x)$ with two layers of size 10 and express the neural network as a linear complementarity problem using Lemma \ref{LCP_ReLU_equivalency}. 

\begin{figure}[t]
	\centering
	\begin{subfigure}[t]{.5 \textwidth}
		\centering
		\includegraphics[width=1\linewidth]{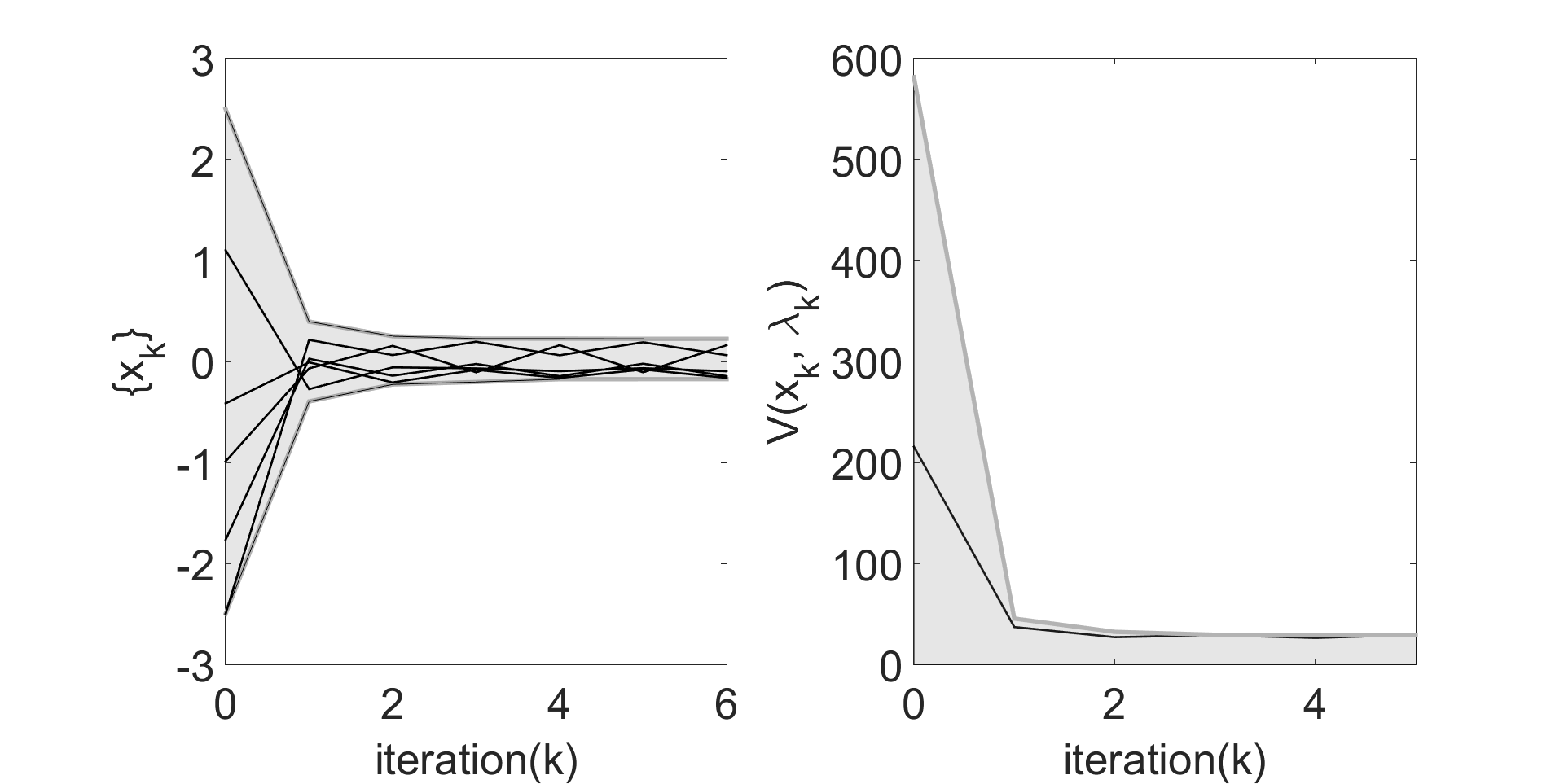}
		\caption{Envelopes for trajectories.}
		\label{fig:lyap_five}
	\end{subfigure}%
	\begin{subfigure}[t]{.5\textwidth}
		\centering
		\includegraphics[width=1\linewidth]{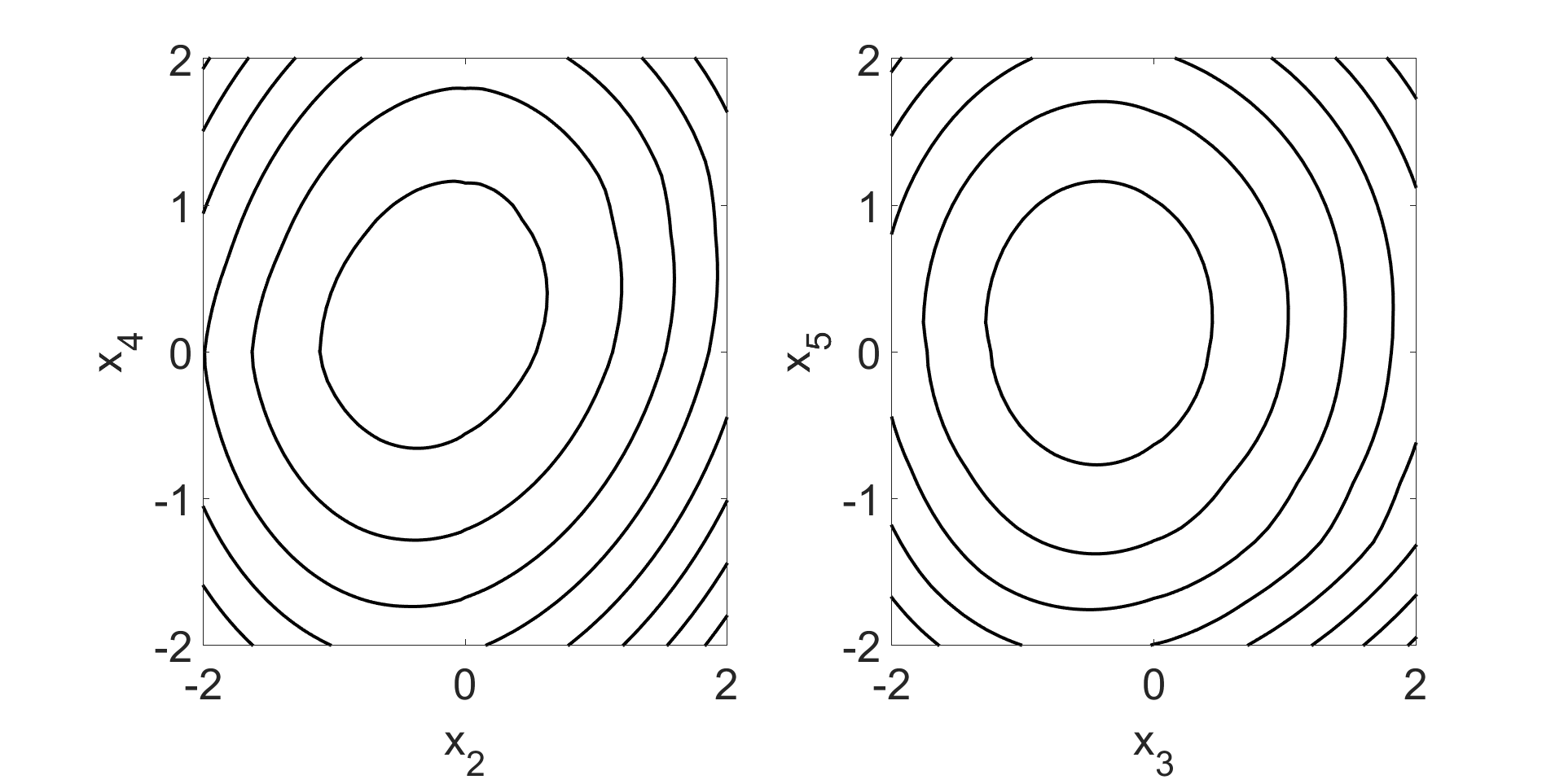}
		\caption{Sublevel sets of the piece-wise quadratic Lyapunov function for the five carts example on the planes $\mathcal{P}_1~=~\{x : x_1 = x_3 = x_5 = 0\}$ and $\mathcal{P}_2 = \{x : x_1 = x_2 = x_4 = 0\}$ respectively.}
		\label{fig:sublevel_five}
	\end{subfigure}
	\caption{Experimental results for the five carts example.}
\end{figure}

We compute a piece-wise quadratic Lyapunov function of the form \eqref{eq:lyapunov_function} in 2.1 seconds (sub-level sets as in Figure \ref{fig:sublevel_five}) that verifies that the closed-loop system with the neural network controller $\phi(x)$ is globally exponentially stable outside a ball $\mathcal{D} = \{ x : ||x||_2^2 > 0.1 \}$. We also verified that there isn't a common Lyapunov function that satisfies the LMI's in (\ref{thm: LMIs}). We note that a common Lyapunov function that satisfies Theorem \ref{stability_theorem} might exist, but no such function satisfies our relaxation in (\ref{thm: LMIs}). On the other hand, this demonstrates the importance of searching over a wider class of functions.
In Figure \ref{fig:lyap_five}, we present the envelopes for 1000 trajectories and the corresponding Lyapunov functions.
We note that \emph{memory} is the limiting factor in terms of scalability of our method and present scalability tests in Table \ref{tab:memory}.

\section{Conclusion and Future Work}

In this work, we have shown that neural networks with ReLU activation functions have an equivalent linear complementarity problem representation. Furthermore, we have shown that a linear complementarity system with a ReLU neural network controller can be transformed into an LCS with a higher dimensional complementarity variable. This allows one to use the existing literature on linear complementarity systems when analyzing an LCS with NN controller. 

\begin{table}[t]
	\centering
	\begin{tabular}{|c|c|c|}
		\hline
		\textbf{RAM}       & \textbf{Number of neurons} & \textbf{Solve time} \\ \hline
		\textit{8GB RAM}   & 20                         & 2.1 seconds         \\ \hline
		\textit{8GB RAM}   & 60                         & 194.72 seconds      \\ \hline
		\textit{8GB RAM}   & 100                        & OOM       \\ \hline
		\textit{16GB RAM}  & 100                        & 1364.78 seconds     \\ \hline
		\textit{16GB RAM}  & 140                        & OOM      \\ \hline
	\end{tabular}
	\caption{Scalability tests.}
	\label{tab:memory}
\end{table}

Towards this direction, we have derived the discrete-time version of the stability results in \cite{camlibel2007lyapunov} and shown that searching for a Lyapunov function for an LCS with ReLU NN controller is equivalent to finding a feasible solution to a set of linear matrix inequalities. The proposed method exploits the complementarity structure of both the system and the NN controller and avoids enumerating the exponential number of potential modes. We have also demonstrated the effectiveness of our method on numerical examples, including a difference inclusion model.

As future work, we are planning to explore tools from algebraic geometry that use samples instead of the S-procedure terms which result in a stronger relaxation \cite{cifuentes2017sampling}. Also, we consider using passivity results \cite{miranda2018dominance} in order to develop algorithms that can verify the stability for larger neural networks. At last, it is of interest to learn stabilizing neural network controllers utilizing the complementarity viewpoint.

\section*{Acknowledgment}
	The authors would like to thank Yike Li (University of
	Pennsylvania) for the code that computes the optimal control sequence for a given LCS. This work was supported by the National Science Foundation under Grant No. CMMI-1830218.

\bibliographystyle{plain}
\bibliography{sample-base,ref_accpm}

\begin{thebibliography}{10}

\bibitem{amos2017optnet}
Brandon Amos and J~Zico Kolter.
\newblock Optnet: Differentiable optimization as a layer in neural networks.
\newblock {\em arXiv preprint arXiv:1703.00443}, 2017.

\bibitem{aydinoglu2020contact}
Alp Aydinoglu, Victor~M Preciado, and Michael Posa.
\newblock Contact-aware controller design for complementarity systems.
\newblock In {\em 2020 IEEE International Conference on Robotics and Automation
  (ICRA)}, pages 1525--1531. IEEE, 2020.

\bibitem{aydinoglu2020stabilization}
Alp Aydinoglu, Victor~M Preciado, and Michael Posa.
\newblock Stabilization of complementarity systems via contact-aware
  controllers.
\newblock {\em arXiv preprint arXiv:2008.02104}, 2020.

\bibitem{bemporad2002explicit}
Alberto Bemporad, Manfred Morari, Vivek Dua, and Efstratios~N Pistikopoulos.
\newblock The explicit linear quadratic regulator for constrained systems.
\newblock {\em Automatica}, 38(1):3--20, 2002.

\bibitem{boyd1994linear}
Stephen Boyd, Laurent El~Ghaoui, Eric Feron, and Venkataramanan Balakrishnan.
\newblock {\em Linear matrix inequalities in system and control theory}.
\newblock SIAM, 1994.

\bibitem{brogliato1999nonsmooth}
Bernard Brogliato.
\newblock {\em Nonsmooth mechanics}.
\newblock Springer, 1999.

\bibitem{camlibel2007lyapunov}
M~Kanat Camlibel, Jong-Shi Pang, and Jinglai Shen.
\newblock Lyapunov stability of complementarity and extended systems.
\newblock {\em SIAM Journal on Optimization}, 17(4):1056--1101, 2007.

\bibitem{chen2020learning}
Shaoru Chen, Mahyar Fazlyab, Manfred Morari, George~J Pappas, and Victor~M
  Preciado.
\newblock Learning lyapunov functions for piecewise affine systems with neural
  network controllers.
\newblock {\em arXiv preprint arXiv:2008.06546}, 2020.

\bibitem{chen2018approximating}
Steven Chen, Kelsey Saulnier, Nikolay Atanasov, Daniel~D Lee, Vijay Kumar,
  George~J Pappas, and Manfred Morari.
\newblock Approximating explicit model predictive control using constrained
  neural networks.
\newblock In {\em 2018 Annual American control conference (ACC)}, pages
  1520--1527. IEEE, 2018.

\bibitem{cifuentes2017sampling}
Diego Cifuentes and Pablo~A Parrilo.
\newblock Sampling algebraic varieties for sum of squares programs.
\newblock {\em SIAM Journal on Optimization}, 27(4):2381--2404, 2017.

\bibitem{cottle2009linear}
Richard~W Cottle, Jong-Shi Pang, and Richard~E Stone.
\newblock {\em The linear complementarity problem}.
\newblock SIAM, 2009.

\bibitem{deits2019lvis}
Robin Deits, Twan Koolen, and Russ Tedrake.
\newblock Lvis: Learning from value function intervals for contact-aware robot
  controllers.
\newblock In {\em 2019 International Conference on Robotics and Automation
  (ICRA)}, pages 7762--7768. IEEE, 2019.

\bibitem{dirkse1995path}
Steven~P Dirkse and Michael~C Ferris.
\newblock The path solver: a nommonotone stabilization scheme for mixed
  complementarity problems.
\newblock {\em Optimization methods and software}, 5(2):123--156, 1995.

\bibitem{dontchev1992difference}
Asen Dontchev and Frank Lempio.
\newblock Difference methods for differential inclusions: a survey.
\newblock {\em SIAM review}, 34(2):263--294, 1992.

\bibitem{fazlyab2019probabilistic}
Mahyar Fazlyab, Manfred Morari, and George~J Pappas.
\newblock Probabilistic verification and reachability analysis of neural
  networks via semidefinite programming.
\newblock In {\em 2019 IEEE 58th Conference on Decision and Control (CDC)},
  pages 2726--2731. IEEE, 2019.

\bibitem{fazlyab2019safety}
Mahyar Fazlyab, Manfred Morari, and George~J Pappas.
\newblock Safety verification and robustness analysis of neural networks via
  quadratic constraints and semidefinite programming.
\newblock {\em arXiv preprint arXiv:1903.01287}, 2019.

\bibitem{fazlyab2019efficient}
Mahyar Fazlyab, Alexander Robey, Hamed Hassani, Manfred Morari, and George
  Pappas.
\newblock Efficient and accurate estimation of lipschitz constants for deep
  neural networks.
\newblock In {\em Advances in Neural Information Processing Systems}, pages
  11427--11438, 2019.

\bibitem{haarnoja2018learning}
Tuomas Haarnoja, Sehoon Ha, Aurick Zhou, Jie Tan, George Tucker, and Sergey
  Levine.
\newblock Learning to walk via deep reinforcement learning.
\newblock {\em arXiv preprint arXiv:1812.11103}, 2018.

\bibitem{halm2018quasi}
Mathew Halm and Michael Posa.
\newblock A quasi-static model and simulation approach for pushing, grasping,
  and jamming.
\newblock In {\em International Workshop on the Algorithmic Foundations of
  Robotics}, pages 491--507. Springer, 2018.

\bibitem{heemels2001equivalence}
Wilhemus~PMH Heemels, Bart De~Schutter, and Alberto Bemporad.
\newblock Equivalence of hybrid dynamical models.
\newblock {\em Automatica}, 37(7):1085--1091, 2001.

\bibitem{heemels2000linear}
WPMH Heemels, Johannes~M Schumacher, and S~Weiland.
\newblock Linear complementarity systems.
\newblock {\em SIAM journal on applied mathematics}, 60(4):1234--1269, 2000.

\bibitem{hertneck2018learning}
Michael Hertneck, Johannes K{\"o}hler, Sebastian Trimpe, and Frank
  Allg{\"o}wer.
\newblock Learning an approximate model predictive controller with guarantees.
\newblock {\em IEEE Control Systems Letters}, 2(3):543--548, 2018.

\bibitem{karg2020efficient}
Benjamin Karg and Sergio Lucia.
\newblock Efficient representation and approximation of model predictive
  control laws via deep learning.
\newblock {\em IEEE Transactions on Cybernetics}, 50(9):3866--3878, 2020.

\bibitem{karg2020stability}
Benjamin Karg and Sergio Lucia.
\newblock Stability and feasibility of neural network-based controllers via
  output range analysis.
\newblock {\em arXiv preprint arXiv:2004.00521}, 2020.

\bibitem{khalil2002nonlinear}
Hassan~K Khalil and Jessy~W Grizzle.
\newblock {\em Nonlinear systems}, volume~3.
\newblock Prentice hall Upper Saddle River, NJ, 2002.

\bibitem{lofberg2004yalmip}
Johan Lofberg.
\newblock Yalmip: A toolbox for modeling and optimization in matlab.
\newblock In {\em 2004 IEEE international conference on robotics and automation
  (IEEE Cat. No. 04CH37508)}, pages 284--289. IEEE, 2004.

\bibitem{marcucci2020warm}
Tobia Marcucci and Russ Tedrake.
\newblock Warm start of mixed-integer programs for model predictive control of
  hybrid systems.
\newblock {\em IEEE Transactions on Automatic Control}, 2020.

\bibitem{miranda2018dominance}
Felix~A Miranda-Villatoro, Fulvio Forni, and Rodolphe Sepulchre.
\newblock Dominance analysis of linear complementarity systems.
\newblock {\em arXiv preprint arXiv:1802.00284}, 2018.

\bibitem{mosek2010mosek}
APS Mosek.
\newblock The mosek optimization software.
\newblock {\em Online at http://www. mosek. com}, 54(2-1):5, 2010.

\bibitem{o2020operator}
Brendan O'Donoghue.
\newblock Operator splitting for a homogeneous embedding of the monotone linear
  complementarity problem.
\newblock {\em arXiv preprint arXiv:2004.02177}, 2020.

\bibitem{parisini1995receding}
Thomas Parisini and Riccardo Zoppoli.
\newblock A receding-horizon regulator for nonlinear systems and a neural
  approximation.
\newblock {\em Automatica}, 31(10):1443--1451, 1995.

\bibitem{paszke2017automatic}
Adam Paszke, Sam Gross, Soumith Chintala, Gregory Chanan, Edward Yang, Zachary
  DeVito, Zeming Lin, Alban Desmaison, Luca Antiga, and Adam Lerer.
\newblock Automatic differentiation in pytorch.
\newblock 2017.

\bibitem{posa2014direct}
Michael Posa, Cecilia Cantu, and Russ Tedrake.
\newblock A direct method for trajectory optimization of rigid bodies through
  contact.
\newblock {\em The International Journal of Robotics Research}, 33(1):69--81,
  2014.

\bibitem{posa2015stability}
Michael Posa, Mark Tobenkin, and Russ Tedrake.
\newblock Stability analysis and control of rigid-body systems with impacts and
  friction.
\newblock {\em IEEE Transactions on Automatic Control}, 61(6):1423--1437, 2015.

\bibitem{raghunathan2018semidefinite}
Aditi Raghunathan, Jacob Steinhardt, and Percy~S Liang.
\newblock Semidefinite relaxations for certifying robustness to adversarial
  examples.
\newblock In {\em Advances in Neural Information Processing Systems}, pages
  10877--10887, 2018.

\bibitem{scholtes2012introduction}
Stefan Scholtes.
\newblock {\em Introduction to piecewise differentiable equations}.
\newblock Springer Science \& Business Media, 2012.

\bibitem{smirnov2002introduction}
Georgi~V Smirnov.
\newblock {\em Introduction to the theory of differential inclusions},
  volume~41.
\newblock American Mathematical Soc., 2002.

\bibitem{stewart1996implicit}
David~E Stewart and Jeffrey~C Trinkle.
\newblock An implicit time-stepping scheme for rigid body dynamics with
  inelastic collisions and coulomb friction.
\newblock {\em International Journal for Numerical Methods in Engineering},
  39(15):2673--2691, 1996.

\bibitem{xie2018feedback}
Zhaoming Xie, Glen Berseth, Patrick Clary, Jonathan Hurst, and Michiel van~de
  Panne.
\newblock Feedback control for cassie with deep reinforcement learning.
\newblock In {\em 2018 IEEE/RSJ International Conference on Intelligent Robots
  and Systems (IROS)}, pages 1241--1246. IEEE, 2018.

\bibitem{yin2020stability}
He~Yin, Peter Seiler, and Murat Arcak.
\newblock Stability analysis using quadratic constraints for systems with
  neural network controllers.
\newblock {\em arXiv preprint arXiv:2006.07579}, 2020.

\bibitem{zhang2019safe}
Xiaojing Zhang, Monimoy Bujarbaruah, and Francesco Borrelli.
\newblock Safe and near-optimal policy learning for model predictive control
  using primal-dual neural networks.
\newblock In {\em 2019 American Control Conference (ACC)}, pages 354--359.
  IEEE, 2019.

\bibitem{zhu2019dexterous}
Henry Zhu, Abhishek Gupta, Aravind Rajeswaran, Sergey Levine, and Vikash Kumar.
\newblock Dexterous manipulation with deep reinforcement learning: Efficient,
  general, and low-cost.
\newblock In {\em 2019 International Conference on Robotics and Automation
  (ICRA)}, pages 3651--3657. IEEE, 2019.

\end{thebibliography}

\appendix
\section{Optimal Control of LCS}
\label{ap:opt_cont}
Given an LCS \eqref{eq:LCS} and an initial condition $x_0$, the optimal control problem is connecting complementarity constraints into equivalent big-M mixed integer constraints:
\begin{equation}
	\label{eq:opt_original}
	\begin{aligned}
		\min_{x_k, \lambda_k, u_k} \quad & \sum_{k=0}^{N-1} x_k^T Q^\text{OPT} x_t + u_t^T R^\text{OPT} u_t  + x_k^T Q^\text{OPT}_N x_k \\
		\textrm{s.t.} \quad &x_{k+1} = A x_k + B u_k + D\lambda_k + z, \\
		& M_1 s_k \geq E x_k + F \lambda_k + H u_k + c \geq 0, \\
		& M_2 (\mathbf{1} - s_k) \geq \lambda_k \geq 0, \\
		& s_k \in \{ 0,1  \}^m, \\
	\end{aligned}
\end{equation}
where $\mathbf{1}$ is a vector of ones, and $M_1$, $M_2$ are scalars that are used for the big M method. Notice that the optimization problem \eqref{eq:opt_original} is a mixed integer quadratic program and the optimal solution can be found using branch and bound algorithms.

In this work, we consider function $\pi_\text{OPT}(x_0)$ that returns the first element of the optimal input sequence, $u^*_0$, for a given $x_0$ and learn this function using a neural network $\phi(x)$.

\section{MPC Controller For LTI Systems}
Given an initial condition $x_0$ and an LTI system:
\begin{equation*}
	x_{k+1} = A x_k + B u_k,
\end{equation*}
we consider the following optimal control problem:
\begin{equation}
	\label{eq:MPC_original}
	\begin{aligned}
		\min_{x_k, u_k} \quad & \sum_{k=0}^N x_k^T Q^\text{MPC} x_t + u_t^T R^\text{MPC} u_t  \\
		\textrm{s.t.} \quad &x_{k+1} = A x_k + B u_k, \\
		& x \in \mathcal{X}, u \in \mathcal{U}, \\
	\end{aligned}
\end{equation}
where $\mathcal{X}$ and $\mathcal{U}$ are convex sets that represent the state and input constraints respectively. In this work, we consider the function $\pi_{MPC}(x_0)$ that returns the first element of the optimal input sequence $u_0^*$ for a given $x_0$. We approximate the function $\pi_\text{MPC}(x)$ with a neural network $\phi(x)$.

\end{document}